\documentclass[english,10pt,journal]{IEEEtran}

\usepackage[T1]{fontenc}
\usepackage[latin9]{inputenc}
\usepackage{amsmath}
\usepackage{amssymb}
\usepackage{graphicx}
\PassOptionsToPackage{normalem}{ulem}
\usepackage{ulem}

\IEEEoverridecommandlockouts

\makeatletter

\newcommand{\lyxdot}{.}

\usepackage{tabularx}

\usepackage[ruled, lined, noend, linesnumbered]{algorithm2e}
\SetKwBlock{RepeatNoEnd}{repeat}{}

\usepackage{amsthm}
\usepackage{subfigure}
\usepackage{float}
\usepackage{amsmath}
\usepackage{graphicx}
\usepackage{setspace}
\usepackage{bbm}
\usepackage[shortlabels]{enumitem}

\allowdisplaybreaks

\newtheorem{theorem}{Theorem}

\newtheorem{corollary}{Corollary}

\newtheorem{lemma}{Lemma}
\newtheorem{assumption}{Assumption}

\usepackage{babel}
\usepackage{upquote}
\usepackage{booktabs}\usepackage{bm}\usepackage{cite}
\usepackage{comment}

\newcommand{\expect}[1]{\mathbb{E}\left\{#1\right\}}

\usepackage{color}

\makeatother

\begin{document}

\title{Dynamic Service Migration in Mobile Edge Computing Based on Markov Decision Process
}

\author{\IEEEauthorblockN{Shiqiang Wang, Rahul Urgaonkar, Murtaza Zafer, Ting He, Kevin Chan, Kin K. Leung}
\thanks{
This research was sponsored in part by the U.S. Army Research Laboratory and the U.K. Ministry of Defence and was accomplished under Agreement Number W911NF-06-3-0001 and W911NF-16-3-0001. The views and conclusions contained in this document are those of the author(s) and should not be interpreted as representing the official policies, either expressed or implied, of the U.S. Army Research Laboratory, the U.S. Government, the U.K. Ministry of Defence or the U.K. Government. The U.S. and U.K. Governments are authorized to reproduce and distribute reprints for Government purposes notwithstanding any copyright notation hereon.

S. Wang is with IBM T. J. Watson Research Center, Yorktown Heights, NY, United States, Email: wangshiq@us.ibm.com

R. Urgaonkar is with Amazon Inc., Seattle, WA, United States. Email:
rahul.urgaonkar@gmail.com

M. Zafer is with Nyansa Inc., Palo Alto, CA, United States, Email: murtaza.zafer@gmail.com

T. He is with Pennsylvania State University, University Park, PA, USA. Email: t.he@cse.psu.edu

K. Chan is with Army Research Laboratory, Adelphi, MD, United States, Email: kevin.s.chan.civ@mail.mil

K. K. Leung is with the Department of Electrical and Electronic Engineering, Imperial College London, United Kingdom, Email:  kin.leung@imperial.ac.uk

This paper has been accepted for publication in the IEEE/ACM Transactions on Networking.
A preliminary version of this paper was presented at IFIP Networking 2015~\cite{migrationMain}. Part of this work also appeared in S. Wang's Ph.D. thesis\cite{wang2015dynamic}.
}
\vspace{-0.3in}
}

\maketitle

\begin{abstract}
In mobile edge computing, local edge servers can host cloud-based services, which reduces network overhead and latency but requires service migrations as users move to new locations. It is challenging to make migration decisions optimally because of the uncertainty in such a dynamic cloud environment. In this paper, we formulate the service migration problem as a Markov Decision Process (MDP). Our formulation captures general cost models and provides a mathematical framework to design optimal service migration policies. In order to overcome the complexity associated with computing the optimal policy, we approximate the underlying state space by the distance between the user and service locations. We show that the resulting MDP is exact for uniform one-dimensional user mobility while it provides a close approximation for uniform two-dimensional mobility with a constant additive error. We also propose a new algorithm and a numerical technique for computing the optimal solution which is significantly faster than traditional methods based on standard value or policy iteration. We illustrate the application of our solution in practical scenarios where many theoretical assumptions are relaxed. Our evaluations based on real-world mobility traces of San Francisco taxis show superior performance of the proposed solution compared to baseline solutions. 
\end{abstract}

\section{Introduction}

\label{section:intro}

Mobile applications that utilize cloud computing technologies have become increasingly popular over the recent years, with examples including data streaming, real-time video processing, social networking, etc. Such applications generally consist of a front-end component running on the mobile device and a back-end component running on the cloud \cite{dinh2013survey}, where the cloud provides additional data processing and computational capabilities. With this architecture, it is possible to access complex services from handheld devices that have limited processing power. However, it also introduces new challenges including increased network overhead and access delay to services. 

\emph{Mobile edge computing (MEC)} has recently emerged as a promising technique to address these challenges by moving computation closer to users \cite{SatyanarayananEmergenceEdgeComputing,mach2017survey,mao2017survey}. In MEC, a small number of servers or micro data-centers that can host cloud applications are distributed across the network and connected directly to entities (such as cellular base stations or wireless access points) at the network edge, as shown in Fig. \ref{fig:scenario}. MEC can significantly reduce the service access delay \cite{SatyanarayananEmergenceEdgeComputing}, thereby enabling newly emerging delay-sensitive and data-intensive mobile applications such as augmented reality (AR) and virtual reality (VR)~\cite{zhang2017VR}.
This idea received significant academic and commercial interest recently~\cite{mobileCloudConverge,IBMWhitepaper}.  MEC is also more robust than traditional centralized cloud computing systems~\cite{CloudletHostile}, because the edge servers (ES) are distributed and are thus less impacted by failures at a centralized location.
The MEC concept is also known as cloudlet \cite{CloudletHostile}, edge cloud~\cite{MachenMigration2017}, fog computing~\cite{bonomi2012fog}, follow me cloud~\cite{FollowMeMagazine}, micro cloud \cite{wangPredictedCost2017, SelimiServicePlacement2017}, and small cell cloud~\cite{BecvarPIMRC2014}.

\begin{figure}
\center{\includegraphics[width=0.8\linewidth]{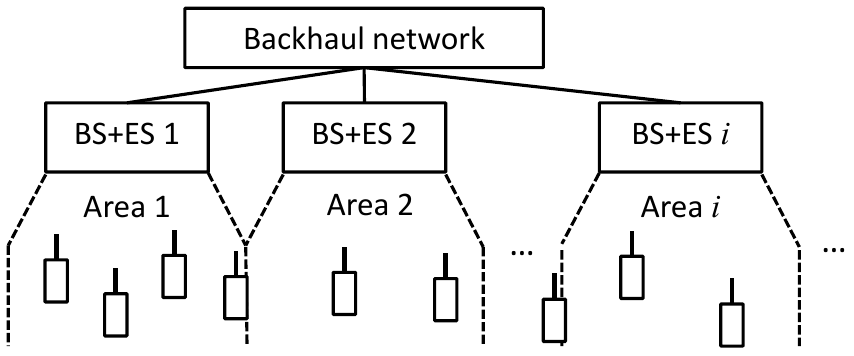}}

\protect\caption{Application scenario of mobile edge computing where  edge servers (ES) are co-located with base stations (BS).}
\label{fig:scenario}
\end{figure}

One new problem that MEC brings in is dynamic service placement and migration. As a user moves across different geographical locations, its service may need to be migrated to follow the user so that the benefits of MEC are maintained. The question is \emph{when and where to migrate the service}. Migrating a service may incur service interruption and network overhead, whereas not migrating a service may increase the data transmission delay between the user and the ES that hosts its service when the user moves away from its original location.
It is challenging to make an optimal migration decision because of the uncertainty in user mobility as well as the complex trade-off between the ``costs'' related to migration and distant data transmission.

The performance of MEC in the presence of user mobility is first studied in \cite{FollowMeGC2013} using a Markovian mobility model, but decisions on whether and where to migrate the service are not considered. A preliminary work on mobility-driven service migration based on Markov Decision Processes (MDPs) is given in \cite{MDPFollowMeICC2014}, which mainly considers one-dimensional (\mbox{1-D}) mobility patterns with a specific cost function. Standard solution procedures are used to solve this MDP, which can be time consuming when the MDP has a large number of states. 
Because the cost functions and transition probabilities of the MDP may change over time and the ES processing power is limited, it is desirable to solve the MDP in an effective manner. 
With this motivation, a more efficient solution to the 1-D mobility case was proposed in \cite{wang2014milcom}, where the transmission and migration costs are assumed to be constant whenever transmission/migration occurs. 
To the best of our knowledge, two-dimensional (\mbox{2-D}) mobility in an MDP setting of the service migration problem has not been considered in the literature, which is a much more realistic case compared to 1-D mobility and we consider it in this paper. 

In this paper, we use the MDP framework to study service migration in MEC. We provide novel contributions beyond \cite{MDPFollowMeICC2014} and \cite{wang2014milcom}, by considering general cost models, 2-D user mobility, and application to real-world mobility traces. 
We focus on the case where ESs are co-located with base stations (BS)\footnote{The notion of base station (BS) in this paper can refer to a cellular tower, a wireless access point, or any physical entity that can have an ES attached to it. We do not distinguish among them for simplicity.} in this paper, which is a possible configuration option according to a recently established MEC specification group~\cite{MECWhitePaper} and this setting has also been proposed for commercial products~\cite{IBMWhitepaper}.
However, our proposed solution is not restricted to such cases and can be easily extended to more general scenarios as long as the costs are location-dependent (see Section \ref{sec:costDefInProbFormulation} for cost definitions).
Our main contributions are summarized as follows.

\begin{enumerate}
\item Our formulation captures general cost models and provides a mathematical framework to design optimal service migration policies.
We note that the resulting problem becomes difficult to solve due to the large state space. In order to overcome this challenge, we propose an approximation of the underlying state space by defining the states as the \emph{distance} between the user and the service locations\footnote{Throughout this paper, we mean by \emph{user location} the location of the BS that the user is associated to.}. This approximation becomes exact for uniform 1-D mobility\footnote{The 1-D mobility is an important practical scenario often encountered in transportation networks, such as vehicles moving along a road.}. We prove several structural properties of the distance-based MDP, which includes a closed-form solution to the discounted sum cost. We leverage these properties to develop an algorithm for computing the optimal policy, which reduces the complexity from $O(N^3)$ (by policy iteration \cite[Chapter 6]{puterman2009markov}) to $O(N^2)$, where the number of states in the distance-based MDP is $N+1$.

\item We show how to use the distance-based MDP to approximate the solution for 2-D mobility models, which allows us to efficiently compute a service migration policy for 2-D mobility. For the uniform 2-D mobility, the approximation error is bounded by a constant. Simulation results comparing our approximation solution to the optimal solution (where the optimal solution is obtained from a 2-D MDP directly) suggest that the proposed approach performs very close to optimal and  obtains the solution much faster.

\item We demonstrate how to apply our algorithms in a practical scenario driven by real mobility traces of taxis in San Francisco which involve multiple users and services. The practical scenario includes realistic factors, e.g., not every BS has an ES attached to it, and each ES can only host a limited number of services. We compare the proposed policy with several baseline strategies that include myopic, never-migrate, and always-migrate policies. It is shown that the proposed approach offers significant gains over these baseline approaches.

\end{enumerate}

The remainder of this paper is organized as follows.  Section~\ref{sec:relatedWork} summarizes the related work. Section~\ref{section:problem_formulation} describes the problem formulation. The distance-based MDP model and its optimal policy is discussed in Section~\ref{section:1D_algorithm}. Section~\ref{section:2D1DApprox} focuses on using the distance-based MDP to solve problems with \mbox{2-D} mobility. Section~\ref{section:RealWorldTraces} discusses the application to real-world scenarios. 
Section~\ref{section:Discussions} provides some additional discussions and
Section~\ref{section:Conclusions} draws conclusions.

\section{Related Work}
\label{sec:relatedWork}

Existing work on service migration focuses on workload and energy patterns that vary slowly \cite{migrationTraditionalCloud1,migrationTraditionalCloud2}. In MEC, user mobility is the driving factor of migration, which varies much faster than parameters in conventional clouds.

Service migration in MEC has been studied in an MDP setting in \cite{MDPFollowMeICC2014,wang2014milcom}, as mentioned earlier.
Besides using the MDP framework, it has also been studied in other settings very recently. The work in \cite{wangPredictedCost2017} relies on a separate prediction module that can predict the future costs of each individual user running its service in every possible ES. A migration mechanism based on user mobility prediction utilizing low level channel information was proposed in \cite{Plachy2016}. Perfect knowledge of user mobility within a given time frame is assumed in \cite{CeselliToN2017}. These approaches are difficult to implement in practice because they require a detailed tracking of every user over a long time duration and may also require access to physical-layer information. In contrast, the approach we propose in this paper only requires knowledge on the number of users at and leaving a BS in each timeslot.

Other work on MEC service migration assumes no knowledge on user mobility, but their applicable scenarios are more limited. In \cite{hou2016asymptotically}, an online algorithm was proposed for service migration from a remote cloud to an ES, where only a single ES is considered and does not apply to cases where users move across areas close to different ESs. An online resource allocation and migration mechanism was proposed in \cite{WangICDCS2017}, where it is assumed that the computational workloads (code and data) are ``fluid'' and can be split up into infinitely small pieces. This assumption generally does not hold in practice, because usually a computer program can only be separated in a small number of ways and needs to have a minimal size of data to run. The work in \cite{urgaonkar2015performance} considers non-realtime applications that allow the queueing of user requests before they are served. This is not applicable for steaming applications such as real-time AR/VR. In this paper, we consider the case where each user continuously accesses its service without queueing, and focus on user mobility and realistic fix-sized computational entities.

Due to the importance of service migration triggered by user and system dynamics in MEC, recent work has also focused on the implementation of service migration mechanisms. Effective migration mechanisms of virtual machines and containers in an MEC environment were proposed in \cite{ha2015adaptive,ma2017efficient}, which did not study the decision making of when/where to migrate and imply an ``always migrate'' mechanism where the service always follows the user. These migration methods can work together with migration decision algorithms, such as the one we propose in this paper, as suggested in \cite{MachenMigration2017}. Other work focuses on developing protocols for MEC service migration \cite{Saurez2016Migration,zhang2017VR}. In \cite{Saurez2016Migration}, a simple thresholding method for migration decision making, which only looks at the system state at the current time (thus ``myopic''), was proposed, while suggesting that other migration decision algorithms can be plugged into their framework as well.
In \cite{zhang2017VR}, a standard MDP approach where the state space is polynomial in the total number of BS and ES was applied for migration decision making, which can become easily intractable when the number of BS and ES is large. It is also mentioned in \cite{zhang2017VR} that it is important to reduce the complexity of finding migration decisions. The complexity of our proposed approach in this paper does \emph{not} depend on the number of BS and ES, and the state space of the MDP in this paper is much smaller than that in \cite{zhang2017VR}.

A related area of work relevant to user mobility studies handover policies in the context of cellular networks \cite{handoverSurvey}. However, the notion of service migration is very different from cellular handover. Handover is usually decided by signal strengths from different BSs, and a handover must occur if a user's signal is no longer provided satisfactorily by its original BS. In the service migration context, a user may continue receiving service from an ES even if it is no longer associated with that BS, because the user can communicate with a remote ES via its currently associated BS and the backhaul network. As a result, the service for a particular user can potentially be placed on any ES, and the service migration problem has a much larger decision space than the cellular handover problem.

\section{Problem Formulation}
\label{section:problem_formulation}

Consider a mobile user in a 2-D geographical area that accesses a cloud-based service hosted on the ESs. The set of possible locations is given by $\mathcal{L}$, where $\mathcal{L}$ is assumed to be finite (but arbitrarily large). We consider a time-slotted model (see Fig. \ref{fig:timing}) where the user's location remains fixed
for the duration of one slot and changes from one slot to the next according to a Markovian mobility model. 
The time-slotted model can be regarded
as a sampled version of a continuous-time model, and the sampling
can be performed either at equal or non-equal intervals over time. In addition, we assume that each location $l \in \mathcal{L}$ is associated with an ES that can host the service for the user (this assumption will be relaxed in Section~\ref{section:RealWorldTraces}). The locations in $\mathcal{L}$ are represented as \mbox{2-D} vectors and there exists a distance metric $\Vert l_1 - l_2 \Vert$ that can be used to calculate the distance between locations $l_1$ and $l_2$. Note that the distance metric may not be Euclidean distance.
An example of this model is a cellular network in which the user's location is considered as the location of its current BS and the ESs are co-located with the BSs. As shown in Section \ref{section:2D1DApprox}, these locations can be represented as 2-D vectors $(i, j)$ with respect to a reference location (represented by  $(0, 0)$) and the distance between any two locations can be calculated in terms of the number of hops to reach from one cell to another cell. We denote the user and service locations at timeslot $t$ as $u(t)$ and $h(t)$ respectively.

\begin{figure}
\center{\includegraphics[width=1\linewidth]{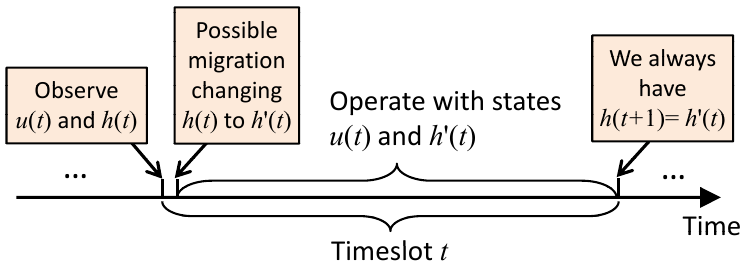}}

\protect\caption{Timing of the proposed service migration mechanism.}
\label{fig:timing}
\end{figure}

\emph{Remark:} 
Although the problem is formulated for the case of a single user accessing a single service, our solution can be applied to manage services for multiple users. We will illustrate such an application in Section~\ref{section:RealWorldTraces}, where we also consider aspects including that ESs are only deployed at a subset of BSs and a limited number of services can be hosted at each ES.
We assume in this paper that different services are independent of each other, and one service serves a single user. The notion of ``service'' can also stand for an instance of a particular type of service, but we do not distinguish between services and instances in this paper for simplicity.

Table \ref{table:notations}   summarizes the main notations in this paper. 

\begin{table}
\small
\renewcommand{\arraystretch}{1.2}
\caption{Summary of main notations}
\label{table:notations}
{\footnotesize
\centering
\begin{tabularx}{\columnwidth}{>{\setlength\hsize{0.3\hsize}\centering}X X}
\hline
Notation & Description\\
\hline
 $\triangleq$ & Is defined to be equal to\\
 $\mathcal{L}$ & Set of locations \\
 $l$ & Location \\
$\Vert l_1 - l_2 \Vert$ & Distance between locations $l_1$ and $l_2$ \\
$u(t)$ & User location at timeslot $t$ \\
$h(t)$ & Service  location at timeslot $t$ \\
${b}(x)$ & Migration cost \\
${c}(y)$ & Transmission cost \\
$s(t)$ & Initial state at slot $t$ \\
$s'(t)$ & Intermediate state at slot $t$ \\
$\pi$ & Decision policy \\
$a(s)$ ($a^*(s)$) & (Optimal) action taken when system is in state $s(t)$ \\
$C_{a}(s)$ & Sum of migration and transmission costs when taking action $a(s)$ in slot~$t$ \\
$V(s_0)$ & Discounted sum cost when starting at state $s_0$ \\
$P{[{s'_0, s_1}]}$ & Transition probability from intermediate state $s'_0$ to the next initial state $s_1$ (in the next slot) \\
$\gamma$ & Discount factor of the MDP \\
$d(t)$ & User-service distance in slot $t$ before possible migration (state in the distance-based MDP)\\
$N$ & Number of states (excluding state zero) in the distance-based MDP \\
$p_0, p, q$ & Transition probabilities of the distance-based MDP (see Fig. \ref{fig:states1D}) \\
$\beta_{c}, \beta_{l}, \delta_{c}, $ $\delta_{l}, \mu, \theta$ & Parameters related to the constant-plus-exponential cost function (see (\ref{eq:costMigration}) and (\ref{eq:costTransmission})) \\
$A_k, B_k, D,$ $ H, m_1, m_2$ & Parameters related to the closed-form solution of the distance-based MDP (see (\ref{eq:finalSolution})--(\ref{eq:diffEquHConst})) \\
$\{n_{k}:k\geq0\}$ & Series of migration states
(i.e., all $n_k$ such that $a(n_{k})\neq n_{k}$) \\
$r$ & Transition probability to one of the neighbors in the 2-D model \\
$e(t)$ & Offset of the user from the service as a 2-D vector (state in the 2-D offset-based MDP)\\
\hline
\end{tabularx}
}
\end{table}

\subsection{Control Decisions and Costs}
\label{sec:costDefInProbFormulation}
At the beginning of each slot, the MEC controller can choose from one of the following control options:
\begin{enumerate}
\item Migrate the service from location $h(t)$ to some other location $h'(t) \in \mathcal{L}$. This incurs a \emph{migration cost} ${b}(x)$ that is assumed to be a 
non-decreasing function of $x$, where $x$ is the distance between $h(t)$ and $h'(t)$, i.e., $x=\Vert h(t)-h'(t)\Vert $. Once the migration is completed, the system operates under state $(u(t), h'(t))$.
The migration cost can capture the service interruption time of the migration process, as recent experimental work has shown that a non-zero interruption time exists whenever a migration occurs (i.e., the migration distance is non-zero) \cite{ha2015adaptive,ma2017efficient,MachenMigration2017}. The interruption time can increase with the migration distance due to increased propagation and switching delays of data transmission.
\item Do not migrate the service. In this case, we have $h'(t)=h(t)$ and the migration cost is ${b}(0)=0$.
\end{enumerate}

In addition to the migration cost, there is a data \emph{transmission cost} incurred by the user for connecting to the currently active service instance. The transmission cost is related to the distance between the service and the user after possible migration, and it is defined as a general non-decreasing function ${c}(y)$, where $y=\Vert u(t) - h'(t) \Vert$. The transmission cost can capture the delay of data transmission, where a high delay increases the service response time. As discussed in \cite{SatyanarayananEmergenceEdgeComputing,CeselliToN2017,ha2015adaptive,XuCloudlet2016}, the delay is usually a function of the geographical or topological distance between two nodes and it increases with the distance. We set ${c}(0)=0$.

We assume that the transmission delay and the service interruption time caused by migration is much smaller than the length of each timeslot (see Fig. \ref{fig:timing}), thus the costs do not change with the timeslot length.

\subsection{Performance Objective}
Let us denote the overall system state at the beginning of each timeslot (before possible migration) by $s(t) = (u(t), h(t))$. The state $s(t)$ is named as the \emph{initial state} of slot $t$.
Consider any policy\footnote{A policy represents a decision rule that maps a state to a new state while (possibly) incurring a cost.} $\pi$ that makes control decisions based on the state $s(t)$ of the system, and we use $a_{\pi}(s(t))$ to represent the control action taken when the system is in state $s(t)$. This action causes the system to transition to a new \emph{intermediate state} $s'(t)= (u(t), h'(t))=a_{\pi}(s(t))$. We use $C_{a_\pi}(s(t))$ to denote the sum of migration and transmission costs incurred by a control $a_\pi(s(t))$ in slot $t$, and we have $C_{a_\pi}(s(t))={b}(\Vert h(t)-h'(t) \Vert)+{c}(\Vert u(t) - h'(t) \Vert)$. Starting from any initial state $ s(0) = s_0$, the long-term expected \emph{discounted sum cost} incurred by policy $\pi$ is given by
\begin{align}
V_{\pi}(s_0)=\lim_{t \to \infty} \expect{\sum_{\tau=0}^{t}\gamma^{\tau}C_{a_\pi}(s(\tau)) \Bigg|s(0) = s_0}
\label{eq:discountedSumCost}
\end{align}
where $0 < \gamma < 1$ is a discount factor.
Note that we consider deterministic policies in this paper and the expectation is taken over random user locations.

Our objective is to design a control policy that minimizes the long-term expected discounted sum total cost starting from any initial state, i.e.,
\begin{align}
V^*(s_0)=\min_{\pi}  V_{\pi}(s_0) \;\;\; \forall s_0.
\label{eq:objFunc}
\end{align}
This problem falls within the class of MDPs with infinite horizon discounted cost. It is well known that the optimal solution is given by a stationary policy\footnote{A stationary policy is a policy where the same decision rule is used in each timeslot.} and can be obtained as the unique solution to the Bellman's equation~\cite{puterman2009markov}:
\begin{align}
V^*(s_0) = \min_a & \Big\{ C_a(s_0) + \gamma \sum_{s_1 \in \mathcal{L}\times\mathcal{L}} P{[{a(s_0),s_1}]} \cdot V^*(s_1) \Big\}
\label{eq:bellman}
\end{align}
where $P{[{a(s_0),s_1}]}$ denotes the probability of transitioning from state $s'(0)=s'_0=a(s_0)$ to $s(1)=s_1$. Note that the intermediate state $s'(t)$ has no randomness when $s(t)$ and $a(\cdot)$ are given, thus we only consider the transition probability from $s'(t)$ to the next state $s(t+1) =(u(t+1), h'(t)) = (u(t+1), h(t+1))$ in (\ref{eq:bellman}), where we note that we always have $h(t+1)=h'(t)$.

\subsection{Characteristics of Optimal Policy}
We next characterize some structural properties of the optimal solution. The following theorem states that it is not optimal to migrate the service
to a location that is farther away from the user than the current service location, as one would intuitively expect.

\begin{theorem}
\label{theorem:notMigrateToFurther} Let $a^*(s)=(u,h')$ denote the optimal action at any state $s=(u,h)$. Then, we have $\Vert u-h'\Vert \leq \Vert u-h\Vert $. (If the optimal action is not unique, then there exists at least one such optimal action.)
\end{theorem}
\begin{proof}
See Appendix \ref{sec:proofOfTheoremNotMigrateFurther}.
\end{proof}

\begin{corollary}
\label{corollary:constant_cost} If ${b}(x)$ and ${c}(y)$ are both \emph{constants} (possibly of different values) for $x,y>0$, and ${b}(0)<{b}(x)$ and ${c}(0)<{c}(y)$ for $x,y>0$, then migrating to locations other than the current location of the mobile user is not optimal.
\end{corollary}
\begin{proof}
See Appendix \ref{sec:proofOfCorollaryConstantCost}.
\end{proof}

\subsection{Simplifying the Search Space}

Theorem \ref{theorem:notMigrateToFurther} simplifies the search space for the optimal policy considerably. However, it is still very challenging to derive the optimal control policy for the general model presented above, particularly when the state space $\{s(t)\}$ is large. One possible approach to address this challenge is to re-define the state space to represent only the \emph{distance} between the user and service locations $d(t)=\Vert u(t) - h(t) \Vert$. The motivation for this comes from the observation that the cost functions in our model only depend on the distance. Note that in general, the optimal control actions can be different for two states $s_0$ and $s_1$ that have the same user-service distance. However, it is reasonable to use the distance as an approximation of the state space for many practical scenarios of interest, and this simplification allows us to formulate a far more tractable MDP. We discuss the distance-based MDP in the next section, and show how the results on the distance-based MDP can be applied to 2-D mobility and real-world scenarios in Sections~\ref{section:2D1DApprox} and~\ref{section:RealWorldTraces}.

In the remainder of this paper, where there is no ambiguity, we reuse the notations $P$, $C_a(\cdot)$, $V(\cdot)$, and $a(\cdot)$ to respectively represent transition probabilities, one-timeslot costs, discounted sum costs, and actions of different MDPs.

\section{Optimal Policy for Distance-Based MDP}
\label{section:1D_algorithm}

In this section, we consider a distance-based\footnote{We assume that the distance is quantized, as it will be the case with the 2-D model discussed in later sections.} MDP where the states $\{d(t)\}$ represent the distances between the user and the service before possible migration (an example is shown in Fig. \ref{fig:states1D}), i.e., $d(t)=\Vert u(t) - h(t) \Vert$. We define the parameter $N$ as an application-specific maximum allowed distance, and we always perform migration when $d(t) \geq N$. We set the actions $a(d(t))=a(N)$ for $d(t)>N$, so that we only need to focus on the states $d(t)\in[0,N]$. After taking action $a(d(t))$, the system operates in the intermediate state $d'(t)=a(d(t))$, and the value of the next state $d(t+1)$ follows the transition probability $P{[{d'(t),d(t+1)}]}$ which is related to the mobility model of the user.
To simplify the solution, we restrict the transition probabilities $P{[{d'(t),d(t+1)}]}$ according to the parameters $p_{0}$, $p$, and $q$ as shown in Fig. \ref{fig:states1D}. Such a restriction is sufficient when the underlying mobility model is a uniform 1-D random walk where the user moves one step to the left or right with equal probability $r_1$ and stays in the same location with probability $1-2r_1$, in which case we can set $p=q=r_1$ and $p_0=2r_1$. This model is also sufficient to approximate the uniform 2-D random walk model, as will be discussed in Section \ref{sub:approxMethodDescription}.

\begin{figure}
\center{\includegraphics[width=1\linewidth]{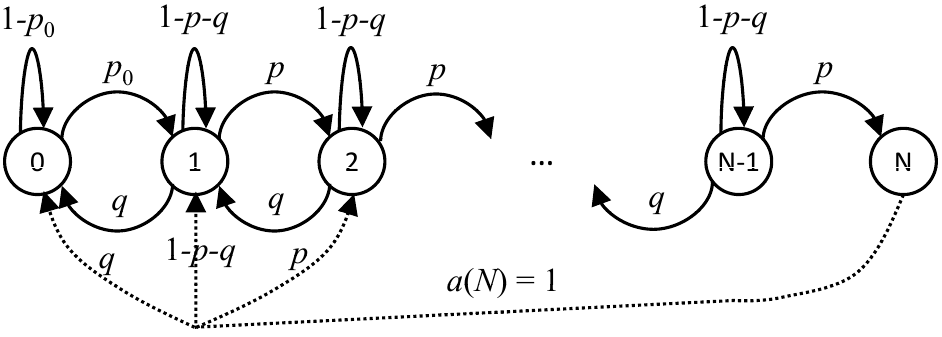}}

\protect\caption{Distance-based MDP with the distances $\{d(t)\}$ (before possible migration) as states. In this example, migration is only performed at state $N$, and only the possible action of $a(N)=1$ is shown for compactness. The solid lines denote state transitions without migration.}
\label{fig:states1D}
\end{figure}

For an action of $d'(t)=a(d(t))$, the new service location $h'(t)$ is chosen such that $x=\Vert h(t)-h'(t)\Vert =|d(t)-d'(t)|$ and $y=\Vert u(t) - h'(t) \Vert = d'(t)$. This means that migration happens along the shortest path that connects $u(t)$ and $h(t)$, and $h'(t)$ is on this shortest path (also note that $d'(t)\leq d(t)$ according to Theorem \ref{theorem:notMigrateToFurther}). Such a migration is possible for the 1-D case where $u(t)$, $h(t)$, and $h'(t)$ are all scalar values. It is also possible for the 2-D case if the distance metric is properly defined (see Section \ref{sub:approxMethodDescription}). The one-timeslot cost is then $C_{a}(d(t))={b}(| d(t)-d'(t)|)+{c}(d'(t))$.

\subsection{Constant-Plus-Exponential Cost Functions}

\label{sub:expCostDef}

To simplify the analysis later, we define the cost functions ${b}(x)$ and ${c}(y)$ in a constant-plus-exponential form:
\begin{equation}
\label{eq:costMigration}
{b}(x)=\begin{cases}
0, & \textrm{if }x=0\\
\beta_{c}+\beta_{l}\mu^{x}, & \textrm{if }x>0
\end{cases}
\end{equation}
\begin{equation}
\label{eq:costTransmission}
{c}(y)=\begin{cases}
0, & \textrm{if }y=0\\
\delta_{c}+\delta_{l}\theta^{y}, & \textrm{if }y>0
\end{cases}
\end{equation}
where $\beta_{c}$, $\beta_{l}$, $\delta_{c}$, $\delta_{l}$, $\mu$, and $\theta$ are real-valued parameters.

The cost functions defined above can have different shapes and are thus applicable to many realistic scenarios (see Fig.~\ref{fig:expCostExample}).
They can approximate an arbitrary cost function as discussed in Appendix~\ref{supSec:approxCostFunc}. 
For example, they can be defined such that there is a constant non-zero cost whenever the distance is larger than zero. Such a cost definition is applicable in systems where all ESs are connected through a single network hub, and it can also approximate cases where there is a relatively high cost whenever the distance larger than zero. The latter is found from experiments in \cite{ha2015adaptive}. The costs can also be defined in a way so that they are (approximately) linearly proportional to the distance, where the distance can be defined as the length of the shortest path between BSs, as in \cite{CeselliToN2017,XuCloudlet2016}.
They also have nice properties allowing us to obtain a closed-form solution to the discounted sum cost, based on which we design an efficient algorithm for finding the optimal policy.

\begin{figure}
\center{\includegraphics[width=0.8\linewidth]{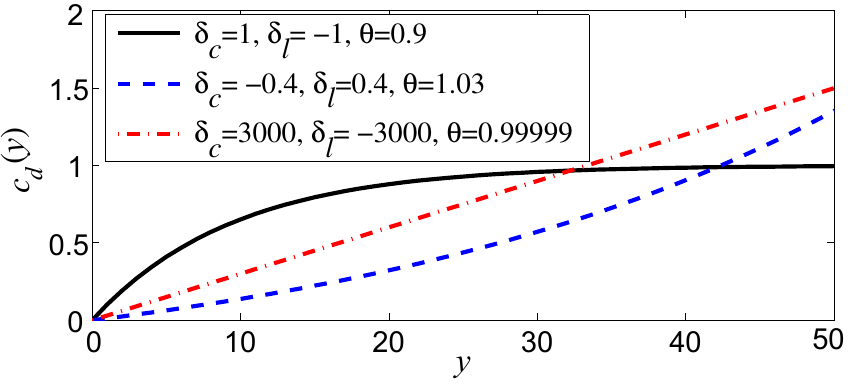}}
\protect\caption{Example of constant-plus-exponential cost function ${c}(y)$.}
\label{fig:expCostExample}
\end{figure}

The parameters $\beta_{c}$, $\beta_{l}$, $\delta_{c}$, $\delta_{l}$, $\mu$, and $\theta$
are selected such that ${b}(x)\geq0$, ${c}(y)\geq0$, and both
${b}(x)$ and ${c}(y)$ are non-decreasing respectively in $x$ and $y$ for $x,y\geq0$. 
Explicitly, we have $\mu\geq 0$; $\beta_l\leq 0$ when $\mu\leq 1$; $\beta_l\geq 0$ when $\mu\geq 1$; $\beta_{c}\geq-\beta_{l}$; $\theta\geq 0$; $\delta_l\leq 0$ when $\theta\leq 1$; $\delta_l\geq 0$ when $\theta\geq 1$; and $\delta_{c}\geq-\delta_{l}$.
The definition that ${b}(0)={c}(0)=0$ is for convenience, because a non-zero cost for $x=0$ or $y=0$ can be offset by the values of $\beta_{c}$ and $\delta_{c}$, thus setting ${b}(0)={c}(0)=0$ does not affect the optimal decision.

With this definition, the values of $\beta_c+\beta_l$ and $\delta_c+\delta_l$ can be regarded as constant terms of the costs, at least such an amount of cost is incurred when $x>0$ and $y>0$, respectively. The parameters $\mu$ and $\theta$ specify the impact of the distance $x$ and $y$, respectively, to the costs, and their values can be related to the network topology and routing mechanism of the network. The parameters $\beta_l$ and $\delta_l$ further adjust the costs proportionally.

\subsection{Closed-Form Solution to Discounted Sum Cost}

\subsubsection{Problem Formulation with Difference Equations}
From (\ref{eq:discountedSumCost}), we can get the following balance equation on the discounted sum cost for a given policy $\pi$:
\begin{align}
V(d) = C_{a}(d) + \gamma  \sum_{d(1) ={a}(d)-1}^{{a}(d)+1}  P{[{{a}(d),d(1)}]} \cdot V(d(1))
\label{eq:balanceDiscountedSumCost}
\end{align}
where we omit the subscript $\pi$ and write $d(0)$ as $d$ for short (we will also follow this convention in the following).

\begin{theorem}
\label{prop:diffEquSolution}
For a given policy $\pi$, let $\{n_{k}:k\geq0\}$ denote the series of all migration states
(such that $a(n_{k})\neq n_{k}$) as specified by policy $\pi$, where $0\leq n_{k}\leq N$.
The discounted sum cost $V(d)$ for $d\in[n_{k-1}, n_k]$ (where we define $n_{-1}\triangleq 0$ for convenience) when following policy $\pi$ can be expressed as
\begin{equation}
V(d)\!=\!A_{k}m_{1}^{d}+B_{k}m_{2}^{d}+D+\begin{cases}
H\cdot\theta^{d} &\!\! \textrm{if }1-\frac{\phi_{1}}{\theta}-\phi_{2}\theta\neq0\\
Hd\cdot\theta^{d} &\!\! \textrm{if }1-\frac{\phi_{1}}{\theta}-\phi_{2}\theta=0
\end{cases}\label{eq:finalSolution}
\end{equation}
where $A_{k}$ and $B_{k}$ are constants corresponding to the interval $[n_{k-1},n_{k}]$, the coefficients $m_1$, $m_2$, $D$, and $H$ are expressed as
\begin{equation}
m_{1}=\frac{1+\sqrt{1-4\phi_{1}\phi_{2}}}{2\phi_{2}},
m_{2}=\frac{1-\sqrt{1-4\phi_{1}\phi_{2}}}{2\phi_{2}}
\label{eq:diffEquMConst}
\end{equation}
\begin{equation}
D={\phi_{3}}\big/({1-\phi_{1}-\phi_{2}})
\label{eq:diffEquDConst}
\end{equation}
\begin{equation}
H=\begin{cases}
\frac{\phi_{4}}{1-\frac{\phi_{1}}{\theta}-\phi_{2}\theta} & \textrm{if }1-\frac{\phi_{1}}{\theta}-\phi_{2}\theta\neq0\\
\frac{\phi_{4}}{\frac{\phi_{1}}{\theta}-\phi_{2}\theta} & \textrm{if }1-\frac{\phi_{1}}{\theta}-\phi_{2}\theta=0
\end{cases}
\label{eq:diffEquHConst}
\end{equation}
where we define $\phi_{1}\triangleq\frac{\gamma q}{1-\gamma(1-p-q)}$, $\phi_{2}\triangleq\frac{\gamma p}{1-\gamma(1-p-q)}$,
$\phi_{3}\triangleq\frac{\delta_{c}}{1-\gamma(1-p-q)}$, and $\phi_{4}\triangleq\frac{\delta_{l}}{1-\gamma(1-p-q)}$.
\end{theorem}

\begin{proof} 
The proof is based on solving a difference equation~\cite{elaydi2005introductionDiffEqu} according to (\ref{eq:balanceDiscountedSumCost}), see Appendix \ref{sec:proofOfDiffEquSolution} for details.
\end{proof} 

We also note that for two different states $d_{1}$ and $d_{2}$, if  policy $\pi$ has actions $a(d_{1})=d_{2}$ and $a(d_{2})=d_{2}$,
then 
\begin{equation}
V(d_{1})={b}\left(\left|d_{1}-d_{2}\right|\right)+V(d_{2})\,.
\label{eq:costRelationship}
\end{equation}

\subsubsection{Finding the Coefficients}

The coefficients $A_{k}$ and $B_{k}$ are unknowns in the solution (\ref{eq:finalSolution}) that need to be
found using additional constraints. Their values may be different for different $k$. After $A_{k}$ and $B_{k}$ are
determined, (\ref{eq:finalSolution}) holds for all $d\in [0,N]$.

We assume $1-\frac{\phi_{1}}{\theta}-\phi_{2}\theta\neq0$ and
$1-\frac{\phi_{2}}{\theta}-\phi_{1}\theta\neq0$ in the following,
the other cases can be derived in a similar way and are omitted for brevity.

\emph{Coefficients for interval $ [0, n_{0}]$:} We have one constraint from the balance
equation (\ref{eq:balanceDiscountedSumCost}) for $d=0$, which is
\begin{equation}
\label{eq:constraints_n0_1_1}
V(0)=\gamma p_{0}V(1)+\gamma(1-p_{0})V(0).
\end{equation}
By substituting (\ref{eq:finalSolution}) into (\ref{eq:constraints_n0_1_1}), we get
\begin{equation}
A_{0}\!\left(1-\phi_{0}m_{1}\right)+B_{0}\!\left(1-\phi_{0}m_{2}\right)\!=\!D\!\left(\phi_{0}-1\right)+H\!\left(\phi_{0}\theta-1\right)\label{eq:constraints_n0_1}
\end{equation}
where $\phi_{0}\triangleq\frac{\gamma p_{0}}{1-\gamma\left(1-p_{0}\right)}$.
We have another constraint by substituting (\ref{eq:finalSolution}) into (\ref{eq:costRelationship}), which gives
\begin{align}
 & A_{0}\left(m_{1}^{n_{0}}-m_{1}^{a(n_{0})}\right)+B_{0}\left(m_{2}^{n_{0}}-m_{2}^{a(n_{0})}\right)\nonumber \\
 & =\beta_{c}+\beta_{l}\mu^{n_{0}-a(n_{0})}-H\left(\theta^{n_{0}}-\theta^{a(n_{0})}\right).\label{eq:constraints_n0_2}
\end{align}
We can find $A_{0}$ and $B_{0}$ from (\ref{eq:constraints_n0_1}) and (\ref{eq:constraints_n0_2}).

\emph{Coefficients for interval $ [n_{k-1}, n_{k}]$:}
Assume that we have found $V(d)$ for all $d\leq n_{k-1}$. By letting $d=n_{k-1}$ in (\ref{eq:finalSolution}), we have the first constraint given by
\begin{equation}
A_{k}m_{1}^{n_{k-1}}+B_{k}m_{2}^{n_{k-1}}=V(n_{k-1})-D-H\cdot\theta^{n_{k-1}}.\label{eq:constraints_ni_1}
\end{equation}
For the second constraint, we consider two cases. 
If $a(n_{k})\leq n_{k-1}$, then
\begin{align}
 & A_{k}m_{1}^{n_{k}}+B_{k}m_{2}^{n_{k}}\nonumber \\
 & =\beta_{c}+\beta_{l}\mu^{n_{k}-a(n_{k})}+V(a(n_{k}))-D-H\cdot\theta^{n_{k}}.\label{eq:constraints_ni_2_1}
\end{align}
If $n_{k-1}<a(n_{k})\leq n_{k}-1$, then
\begin{align}
 & A_{k}\left(m_{1}^{n_{k}}-m_{1}^{a(n_{k})}\right)+B_{k}\left(m_{2}^{n_{k}}-m_{2}^{a(n_{k})}\right)\nonumber \\
 & =\beta_{c}+\beta_{l}\mu^{n_{k}-a(n_{k})}-H\left(\theta^{n_{k}}-\theta^{a(n_{k})}\right).\label{eq:constraints_ni_2_2}
\end{align}
The values of $A_{k}$ and $B_{k}$ can be solved from (\ref{eq:constraints_ni_1}) together with either (\ref{eq:constraints_ni_2_1}) or (\ref{eq:constraints_ni_2_2}).

\subsubsection{Solution is in Closed-Form}
We note that $A_0$ and $B_0$ can be expressed in closed-form, and $A_k$ and $B_k$ for all $k$ can also be expressed in closed-form by substituting (\ref{eq:finalSolution}) into (\ref{eq:constraints_ni_1}) and (\ref{eq:constraints_ni_2_1}) where needed. Thus,  (\ref{eq:finalSolution}) is a \emph{closed-form solution} for all $d\in[0,N]$. Numerically, we can find $V(d)$ for all $d\in[0,N]$ in $O(N)$ time.

\subsection{Algorithm for Finding the Optimal Policy}

Standard approaches of solving for the optimal policy
of an MDP include value iteration and policy iteration \cite[Chapter 6]{puterman2009markov}. Value iteration finds the optimal policy from the Bellman's equation (\ref{eq:bellman}) iteratively, which may require a large number of iterations before converging to the optimal result. Policy iteration generally requires a smaller number of iterations, because, in each iteration, it finds the
exact values of the discounted sum cost $V(d)$  for the policy resulting from the previous
iteration, and performs the iteration based on the exact $V(d)$ values.
However, in general, the $V(d)$ values are found by solving
a system of linear equations, which has a complexity of $O(N^{3})$
when using Gaussian-elimination.

We propose a modified policy-iteration approach for finding the optimal policy, which uses the above result instead of Gaussian-elimination
to compute $V(d)$, and also only checks for migrating to lower states
or not migrating (according to Theorem \ref{theorem:notMigrateToFurther}). The algorithm is shown in Algorithm \ref{alg:modifiedPolicyIteration}, where Lines \ref{algStartFind_ni}--\ref{algEndFind_ni} find the values of $n_k$, Lines \ref{algStartFind_Values}--\ref{algEndFind_Values} find the discounted sum cost values, and Lines \ref{algStartFindOptPolicy}--\ref{algEndFindOptPolicy} update the optimal policy. 
The overall complexity for each iteration is $O\left(N^{2}\right)$ in Algorithm \ref{alg:modifiedPolicyIteration}, which reduces complexity because standard\footnote{We use the term ``standard'' here to distinguish with the modified
policy iteration mechanism proposed in Algorithm \ref{alg:modifiedPolicyIteration}. } policy iteration has complexity $O(N^{3})$, and the standard value iteration approach does not compute the exact value function in each iteration and generally has long convergence time.

\begin{algorithm}[t]
\caption{Modified policy-iteration algorithm based on difference equations}
\label{alg:modifiedPolicyIteration}
{\footnotesize

Initialize $a(d) \leftarrow 0$ for all $d={0,1,2,...,N}$;

Find constants $\phi_0$, $\phi_1$, $\phi_2$, $\phi_3$, $\phi_4$, $m_1$, $m_2$, $D$, and $H$;

\Repeat{$a_\textrm{prev}(d)=a(d)$ for all $d$}{ \label{AlgLargeLoopStart}

	$k \leftarrow 0$; \label{algStartFind_ni}

	\For {$d=1...N$} {
		\If {$a(d)\neq d$}{
			$n_k \leftarrow d$, $k \leftarrow k+1$;
		}
	} \label{algEndFind_ni}

	\For {all $n_k$}{  \label{algStartFind_Values}

		\If {$k=0$}{
			Solve for $A_{0}$ and $B_{0}$ from (\ref{eq:constraints_n0_1}) and (\ref{eq:constraints_n0_2});	
			
			Find $V(d)$ with $0\leq d\leq n_k$ from (\ref{eq:finalSolution}) with $A_{0}$ and $B_{0}$ found above;
		}
		\ElseIf{$k>0$}{
			\If {$a(n_{k})\leq n_{k-1}$}{
				Solve for $A_{k}$ and $B_{k}$ from (\ref{eq:constraints_ni_1}) and (\ref{eq:constraints_ni_2_1});
			}
			\Else{
				Solve for $A_{k}$ and $B_{k}$ from (\ref{eq:constraints_ni_1}) and (\ref{eq:constraints_ni_2_2});
			}
			Find $V(d)$ with $n_{k-1}<d\leq n_k$ from (\ref{eq:finalSolution}) with $A_{k}$ and $B_{k}$ found above;
		}

	} \label{algEndFind_Values}

	\For {$d=1...N$}{ \label{algStartFindOptPolicy}
		$a_\textrm{prev}(d) \leftarrow a(d)$;

		$a(d) \leftarrow \arg\min_{a\leq d}\left\{ C_{a}(d)+\gamma\sum_{j=a-1}^{a+1}P{[{a,j}]}\cdot V(j)\right\} $;
	} \label{algEndFindOptPolicy}

}
\Return $a^*(d) \leftarrow a(d)$ for all $d$;

}
\end{algorithm}

\section{Approximate Solution for 2-D Mobility}
\label{section:2D1DApprox}

In this section, we show that the distance-based MDP can be used to find a near-optimal service migration policy, where the user conforms to a uniform 2-D random walk mobility model in an infinite space. This mobility model can be used to approximate real-world mobility traces (see Section \ref{section:RealWorldTraces}). We consider a hexagonal cell structure, but the approximation procedure can be extended to other \mbox{2-D} mobility models (such as Manhattan grid) after modifications. The user is assumed to transition to one of its six neighbors at the beginning of each timeslot with probability $r$, and stay in the same cell with probability $1-6r$.

\subsection{Offset-Based MDP}
\label{sub:OffsetBasedMDP}

Define the \emph{offset} of the user from the service as a 2-D vector $e(t)=u(t)-h(t)$ (recall that $u(t)$ and $h(t)$ are also 2-D vectors). Due to the space-homogeneity of the mobility model, it is sufficient to model the state of the MDP by $e(t)$ rather than $s(t)$. The distance metric $\Vert l_1-l_2 \Vert$ is defined as the minimum number of hops that are needed to reach from cell $l_1$ to cell $l_2$ on the hexagon model.

\begin{figure}
\center{\includegraphics[width=0.47\linewidth]{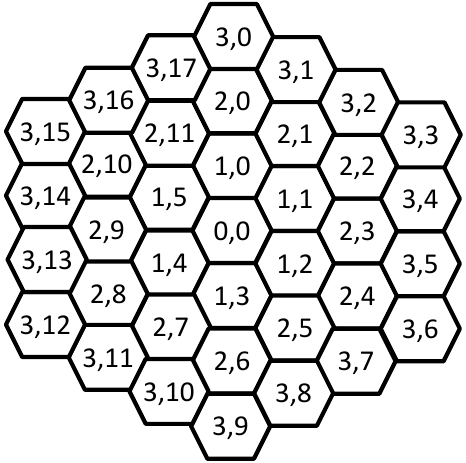}}

\protect\caption{Example of 2-D offset model on hexagon cells, where $N=3$.}
\label{fig:hexCell} 
\end{figure}

We name the states with the same value of $\Vert e(t) \Vert$ as a \emph{ring}, and express the states $\{e(t)\}$ with polar indices $(i,j)$, where the first index $i$ refers to the ring index, and the second index $j$ refers to each of the states within the ring, as shown in Fig.~\ref{fig:hexCell}. For $e(t)=(i,j)$, we have $\Vert e(t) \Vert=i$. If $u(t)=h(t)$ (i.e., the actual user and service locations (cells) are the same), then we have $e(t)=(0,0)$ and $\Vert e(t) \Vert=0$. 

Similarly as in the distance-based MDP, we assume in the 2-D MDP that we always migrate when $\Vert e(t) \Vert \geq N$, where $N$ is a design parameter, and we only consider the state space $\{e(t)\}$ with $\Vert e(t) \Vert \leq N$. 
The system operates in the intermediate state $e'(t)=u(t)-h'(t)=a(e(t))$ after taking action $a(e(t))$. The next state $e(t+1)$ is determined probabilistically according to the transition probability $P{[{e'(t),e(t+1)}]}$. We have $P{[{e'(t),e(t+1)}]}=1-6r$ when $e(t+1)=e'(t)$; $P{[{e'(t),e(t+1)}]}=r$ when $e(t+1)$ is a neighbor of $e'(t)$; and $P{[{e'(t),e(t+1)}]}=0$ otherwise.
Note that we always have $e(t)-e'(t)=h'(t)-h(t)$, so the one-timeslot cost is
$C_{a}(e(t))={b}(\Vert e(t)-e'(t) \Vert)+{c}(\Vert e'(t) \Vert)$.

We note that, even after simplification with the offset model, the 2-D offset-based MDP has a significantly larger number of states compared with the distance-based MDP, because for a distance-based model with $N$ states (excluding state zero), the 2-D offset model has $M=3N^{2}+3N$ states (excluding state $(0,0)$).
Therefore, we use the distance-based MDP proposed in Section \ref{section:1D_algorithm} to approximate the 2-D offset-based MDP, which significantly reduces the computational time as shown in Section \ref{numEval2D1DApprox}.

\subsection{Approximation by Distance-based MDP}

\label{sub:approxMethodDescription}

In the approximation, the parameters of the distance-based MDP are chosen as $p_{0}=6r$, $p=2.5r$, and $q=1.5r$. The intuition of the parameter choice is that, at state $(i'_0,j'_0)=(0,0)$ in the 2-D MDP, the aggregate probability of transitioning to any state in ring $i_1=1$ is $6r$, so we set $p_0=6r$; at any other state $(i'_0,j'_0)\neq(0,0)$, the aggregate probability of transitioning to any state in the higher ring $i_1=i'_0+1$ is either $2r$ or $3r$, and the aggregate probability of transitioning to any state in the lower ring $i_1=i'_0-1$ is either $r$ or $2r$, so we set $p$ and $q$ to the median value of these transition probabilities.

To find the optimal policy for the 2-D MDP, we first find the optimal policy for the distance-based MDP with the parameters defined above. Then, we map the optimal policy from the distance-based MDP to a policy for the \mbox{2-D} MDP. To explain this mapping, we note that, in the \mbox{2-D} hexagon offset model, there always exists at least one shortest path from any state $(i,j)$ to an arbitrary state in ring $i'$, the length of this shortest path is $|i-i'|$, and each ring between $i$ and $i'$ is traversed once on the shortest path. For example, one shortest path from state $(3,2)$ to ring $i'=1$ is $\{(3,2),(2,1),(1,0)\}$. When the system is in state $(i,j)$ and the optimal action from the distance-based MDP is $a^*(i)=i'$, we perform migration on the shortest path from $(i,j)$ to ring $i'$. If there exist multiple shortest paths, one path is arbitrarily chosen. For example, if $a(3)=2$ in the distance-based MDP, then we have either $a(3,2)=(2,1)$ or $a(3,2)=(2,2)$ in the 2-D MDP. With this mapping, the one-timeslot cost $C_a(d(t))$ for the distance-based MDP and the one-timeslot cost $C_a(e(t))$ for the 2-D MDP are the same, because the migration distances in the distance-based MDP and 2-D MDP are the same (thus same migration cost) and all states in the same ring $i'=\Vert e'(t) \Vert=d'(t)$ have the same transmission cost ${c}(\Vert e'(t) \Vert)={c}(d'(t))$.

\subsection{Bound on Approximation Error}
\label{sub:errorBound}

Error arises from the approximation because the transition probabilities in the distance-based MDP are not exactly the same as that in the 2-D MDP (there is at most a difference of $0.5r$). In this subsection, we study the difference in the discounted sum costs when using the  policy obtained from the distance-based MDP and the true optimal policy for the 2-D MDP. The result is summarized as Theorem \ref{prop:approxErrorBound}.

\begin{theorem}\label{prop:approxErrorBound} Let $V^*_{\textnormal{dist}}(e)$ denote the discounted sum cost when using the policy that is optimal for the distance-based MDP, and let $V^*(e)$ denote the discounted sum cost when using true optimal policy of the 2-D MDP, then we have $V^*_{\textnormal{dist}}(e)-V^*(e)\leq \frac{\gamma r \kappa}{1-\gamma}$ for all $e$,
where $\kappa\triangleq\max_{x}\left\{ {b}\left(x+2\right)-{b}\left(x\right)\right\} $.
\end{theorem}
\begin{proof}  
(Outline) The proof is completed in three steps. First, we modify the states of the 2-D MDP in such a way that the aggregate transition probability from any state $(i'_0,j'_0)\neq(0,0)$ to ring $i_1=i'_0+1$ (correspondingly, $i_1=i'_0-1$) is $2.5r$ (correspondingly, $1.5r$). We assume that we use a given policy on both the original and modified \mbox{2-D} MDPs, and show a bound on the difference in the discounted sum costs for these two MDPs. In the second step, we show that the modified \mbox{2-D} MDP is equivalent to the distance-based MDP. This can be intuitively explained by the reason that the modified 2-D MDP has the same transition probabilities as the distance-based MDP when only considering the ring index $i$, and also, the one-timeslot cost $C_a(e(t))$ only depends on $\Vert e(t)-a(e(t)) \Vert$ and $\Vert a(e(t)) \Vert$, both of which can be determined from the ring indices of $e(t)$ and $a(e(t))$. The third step uses the fact that the optimal policy for the distance-based MDP cannot bring higher discounted sum cost for the distance-based MDP (and hence the modified 2-D MDP) than any other policy. By utilizing the error bound found in the first step twice, we prove the result. For details of the proof, see Appendix~\ref{sec:proofOfApproxErrorBound}.
\end{proof}

The error bound is a constant value when all the related parameters are given.
It increases with $\gamma$. However, the absolute value of the discounted sum cost also increases with $\gamma$, so the relative error can remain low.

\subsection{Numerical Evaluation}
\label{numEval2D1DApprox}
The error bound derived in Section \ref{sub:errorBound} is a worst-case
upper bound of the error. In this subsection, we evaluate the performance
of the proposed approximation method numerically, and focus on the
average performance of the approximation.

\begin{figure*}
\center{\subfigure[]{\includegraphics[width=0.32\linewidth]{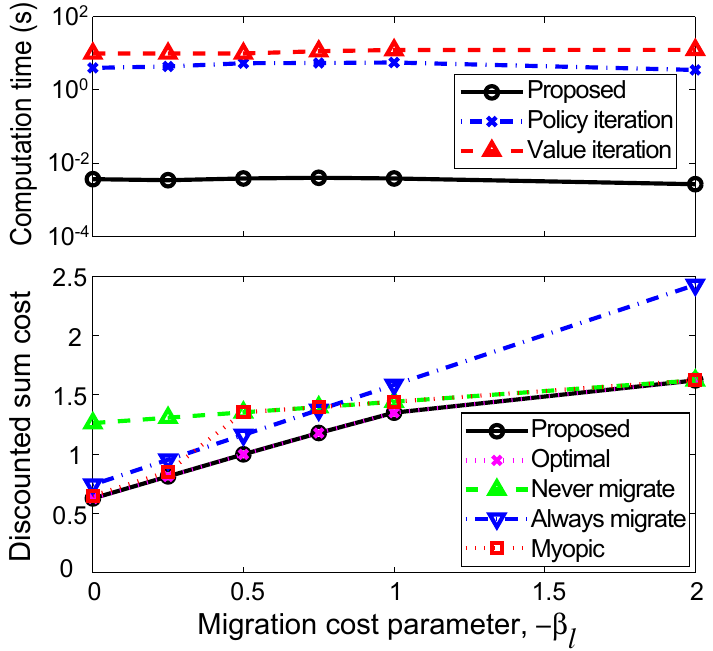}} 
\subfigure[]{\includegraphics[width=0.32\linewidth]{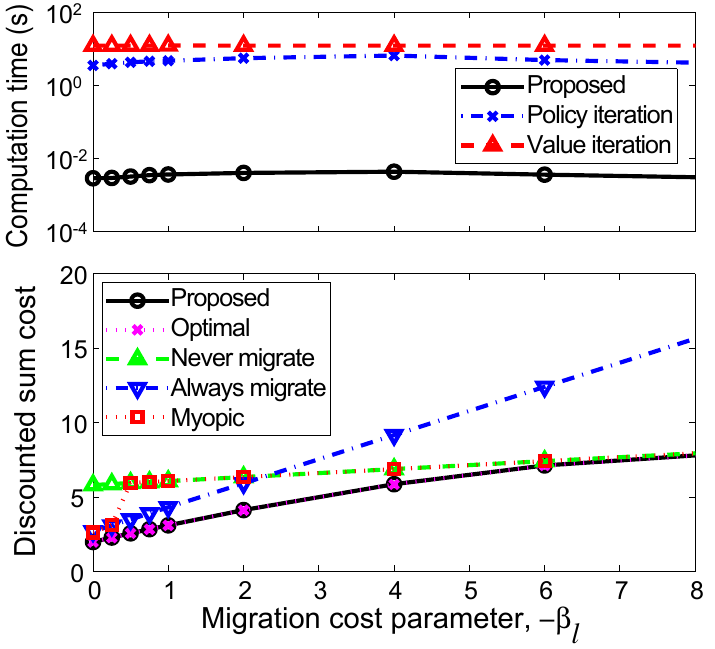}} 
\subfigure[]{\includegraphics[width=0.32\linewidth]{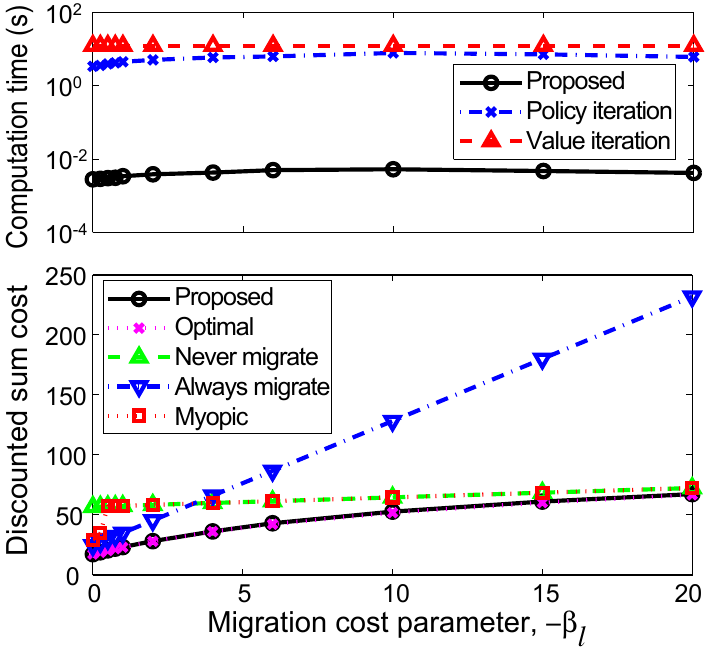}}}

\protect\caption{Simulation result with 2-D random walk: (a) $\gamma=0.5$, (b) $\gamma=0.9$,
(c) $\gamma=0.99$.}
\label{fig:simRandomWalk} 
\vspace{-0.1in}
\end{figure*}

We consider 2-D random walk mobility with randomly chosen parameter $r$. The maximum user-service distance is set as $N=10$. The transmission
cost function parameters are selected as $\theta=0.8$, $\delta_{c}=1$,
and $\delta_{l}=-1$. With these parameters, we have $\delta_{c}+\delta_{l}=0$, which means that there is no constant portion in the cost function.
For the migration cost, we choose $\mu=0.8$ and fix $\beta_{c}+\beta_{l}=1$ to represent a constant server processing cost for migration. The parameter $\beta_{l}\leq 0$
takes different values in the simulations, to represent different
sizes of data to be migrated.

The simulations are performed in MATLAB on a computer with Intel Core
i7-2600 CPU, 8GB memory, and 64-bit Windows 7. We study the computation
time (i.e., the time used to run the algorithm) and the discounted sum cost of the proposed approach that is based on
approximating the original 2-D MDP with the distance-based MDP. For the computation time comparison,
standard value and policy iteration approaches \cite[Chapter 6]{puterman2009markov}
are used to solve the original 2-D MDP. 
The discounted sum cost from the proposed approach is compared with the costs from alternative policies,
including the \emph{true optimal} policy from standard policy iteration on the 2-D model,
the \emph{never-migrate} policy which never migrates except when at states in ring $i\geq N$ (in which case the service is migrated to the current location of the user), the
\emph{always-migrate} policy which always migrates to the current user location when the user and service are at different locations, and the \emph{myopic} policy that chooses actions to minimize the one-slot cost. Note that the always-migrate policy is suggested in \cite{ha2015adaptive,ma2017efficient}, and a myopic policy that makes migration decisions based on instantaneous quality of service observations is proposed in~\cite{Saurez2016Migration}.

The simulations are run with $50$ different random seeds, and the overall
results are shown in Fig. \ref{fig:simRandomWalk} with different
values of the discount factor $\gamma$. 

\subsubsection{Reduction in Computation Time}  Fig. \ref{fig:simRandomWalk} shows that the computation time of the proposed method
is only about $0.1\%$ of that of standard value or
policy iteration. This time reduction is explained as follows. As discussed
in Section \ref{sub:OffsetBasedMDP}, for a distance-based MDP with
$N$ states (excluding state zero), the 2-D MDP has $M=3N^{2}+3N$
states (excluding state $(0,0)$). When we ignore the complexity
of matrix inversion in the policy iteration procedure, the standard
value and policy iteration approaches on the 2-D MDP have a complexity
of $O(M^{2})$ in each iteration, because the optimal action needs
to be found for each state, which requires enumerating through all
the states and all possible actions for each state (similarly as in
Lines \ref{algStartFindOptPolicy}--\ref{algEndFindOptPolicy} of Algorithm
\ref{alg:modifiedPolicyIteration}). In the simulations, $N=10$,
so we have $M=330$. Recall that the complexity of Algorithm~\ref{alg:modifiedPolicyIteration}
used in the proposed approach is $O(N^{2})$, so the ratio of the computational
complexities of different approaches can be approximated by $\frac{M^{2}}{N^{2}}\approx10^{3}$.
Therefore, the standard value and policy iteration  consume
about $10^{3}$ times more computation time compared to the proposed approach. 
Also note that the computation time of our proposed approach is only a few milliseconds (see Fig. \ref{fig:simRandomWalk}), whereas the MDP solution approach used in \cite{zhang2017VR} can easily take several seconds to minutes to obtain the solution for a similar sized system.

\subsubsection{Near-Optimal Cost} 
We can also see from Fig. \ref{fig:simRandomWalk} that the proposed method yields a discounted sum cost that is very close to the optimal cost. The results also provide several insights into the performance of the baseline policies. Specifically, the cost of the always-migrate policy approximates the optimal cost when $|\beta_{l}|$ is small, and the cost of the never-migrate policy approximates the optimal cost when $|\beta_{l}|$ is large. This is because,
when $|\beta_{l}|$ is small, the migration cost is relatively small,
and migration can be beneficial for most cases; when $|\beta_{l}|$
is large, the migration cost is large, and it is better not to migrate
in most cases. The myopic policy is the same as the never-migrate
policy when $|\beta_{l}|\geq0.5$, because ${b}(x)$ and ${c}(y)$ are both concave according to the simulation settings and we always have ${b}(x)\geq {c}(y)$
when $|\beta_{l}|\geq0.5$, where we recall that the myopic policy does not consider the future impact
of actions. There is an intersection of the costs from never-migrate
and always-migrate policies. When $|\beta_{l}|$ takes values that
are larger than the value at the intersection point, the optimal cost
is close to the cost from the never-migrate policy when $\gamma$
is small, and the gap becomes larger with a larger $\gamma$. The
reason is that the benefit of migration becomes more significant when we look
farther ahead into the future. 
We also note that the cost of never-migrate policy slightly increases as $|\beta_{l}|$ increases, because the never-migrate policy also occasionally migrates when the user-service distance greater than or equal to $N$ (see earlier definition).

\section{Application to Real-World Scenarios}
\label{section:RealWorldTraces}

In this section, we discuss how the aforementioned approaches can
be applied to service migration in the real world, where \emph{multiple users and services co-exist} in the cloud system, and the transition probability parameter $r$ is estimated based on an interval of recent observations before the current time.  
We note that in practical scenarios, ESs may \emph{not} be deployed at every BS, and each ES may have a \emph{capacity limit} that restricts the number of services it can host. Theoretically, it is still possible to formulate the service migration problem with these additional constraints as an MDP. However, the resulting MDP will have a significantly larger state space than our current model, and it is far more difficult to solve or approximate this new MDP. While we leave the theoretical analysis of this new MDP as future work, we propose a heuristic approach in this section to handle these additional constraints. The proposed approach is largely based on the results and intuitions obtained in previous sections. The 2-D MDP approximation approach proposed in Section~\ref{sub:approxMethodDescription} is used as a subroutine in the proposed scheme, and the distance-based MDP resulting from the approximation is solved using Algorithm~\ref{alg:modifiedPolicyIteration}.

\subsection{Mapping between Real-World and MDP-Model} 
\label{sub:mapRealToMDPTraces}
The mapping between the real-world and the MDP model is discussed as follows. 

\textbf{MEC Controller:} We assume that there exists a control entity which we refer to as the \emph{MEC controller}. The MEC controller does not need to be a separate cloud entity. Rather, it can be a service running at one of the ESs. 

\textbf{Base Stations:} Each BS (which may or may not have an ES attached to it) is assumed to have basic capability of keeping records on arriving and departing users, and performing simple monitoring and computational operations.

\textbf{Timeslots:} The physical time length corresponding to a slot in the MDP is a pre-specified parameter, which is a constant for ease of presentation.
This parameter can be regarded as a protocol parameter, and it is \emph{not} necessary for all BSs to precisely synchronize on individual timeslots. 

\textbf{Transition Probability:} 
The transition probability parameter $r$ is estimated from the sample paths of multiple users, using the procedure described in Section \ref{sub:overallProcedureTraces} below.
We define a window length $T_{w}\geq 1$ (represented as the number of timeslots), which specifies the amount
of timeslots to look back to estimate the parameter $r$. We consider the
case where $r$ is the same across the whole geographical
area, which is a reasonable assumption when different locations within
the geographical area under consideration have similarities (for example,
they all belong to an urban area). More sophisticated cases can be studied
in the future. 

\textbf{Distance:} 
The discrete distance in the MDP model can be measured using metrics related to the displacement of the user, such as the number of hops between different BSs or quantized Euclidean distance. The simulations in Section~\ref{sub:realWorldSimulation} show that both metrics are good for the proposed method to work well in practical scenarios.

\textbf{Cost Parameters:} The cost parameters
$\beta_{c}$, $\beta_{l}$, $\mu$, $\delta_{c}$, $\delta_{l}$,
and $\theta$ are selected based on the actual application scenario and MEC system characteristics,
and their values may vary with the background traffic load of the network and
ESs. 

\textbf{Discount Factor:} 
The discount factor $\gamma$ balances the trade-off between short-term and long-term costs. For example, if we set $\gamma = 0$, the algorithm minimizes the instantaneous cost only, without considering the future impact. This is essentially the same as the myopic policy. If we set $\gamma \approx 1.0$, then the algorithm aims at minimizing the long-term average cost. However, the instantaneous cost may be high in some timeslots in this case. The choice of $\gamma$ in practice depends on the acceptable level of fluctuation in the instantaneous cost, the importance of low average cost, and the time duration that the user accesses the service. In general, a larger $\gamma$ reduces the long-term average cost but the instantaneous cost in some slots may be higher. Also, if a user only requires the service for a short time, then there is no need to consider the cost for the long-term future. 
For ease of presentation, we set $\gamma$ as a constant value. In practice, the value of $\gamma$ can be different for different users or services.

\textbf{Policy Update Interval:} A policy update interval $T_{u}$ is defined (represented as the number of timeslots), at which a new migration policy is computed by the MEC controller.

\subsection{Overall Procedure} 
\label{sub:overallProcedureTraces}

The data collection, estimation, and service placement procedure is described below. 

\begin{enumerate}[font=\bfseries,align=left, leftmargin=0pt, labelindent=0pt,listparindent=0pt, labelwidth=!, itemindent=!]

\item At the beginning of each slot, the following is performed:

\begin{enumerate}[font=\bfseries,align=left, leftmargin=5pt, labelindent=5pt,listparindent=0pt, labelwidth=!, itemindent=!]

\item Each BS obtains the identities of its associated users. 
Based on this information, the BS computes the number of users that have left the cell (compared to the beginning of the previous timeslot) and the total number of users that are currently in the cell. This information is saved for each timeslot for the duration of $T_w$ and will be used in step \ref{traceOverallProc:computeEmpiricalProb}. 

\item The MEC controller collects information on currently active services on each ES, computes the new placement of services according to the procedure described in Section \ref{sec:actualServicePlacement}, and sends the resulting placement instructions to each ES. The placements of all services are updated based on these instructions.
\label{traceOverallProc:placementDecision}

\end{enumerate}

\item At every interval $T_{u}$, the following is performed:

\begin{enumerate}[font=\bfseries,align=left, leftmargin=5pt, labelindent=5pt,listparindent=0pt, labelwidth=!, itemindent=!]
\item The MEC controller sends a request to all BSs to collect the current statistics. 

\item After receiving the request, each BS $n$ computes the empirical probability of users moving outside of the cell: 
\begin{equation}
f_{n}=\frac{1}{T_w}\sum_{\tau=t-T_w}^{t-1}\frac{m'_{n}(\tau)}{m_{n}(\tau)}
\label{eq:realWorldEst1}
\end{equation}
where the total number of users that are associated to BS $n$ in slot $\tau$ is $m_{n}(\tau)$, among which $m'_{n}(\tau)$ users have disconnected from BS $n$ at the end of slot $\tau$ and these users are associated to a different BS in slot $\tau+1$; and $t$ denotes the current timeslot index.
These empirical probabilities are sent together with other monitored information, such as the current load of the network and ES (if the BS has an ES attached to it), to the MEC controller. 
\label{traceOverallProc:computeEmpiricalProb}

\item After the controller receives responses from all BSs, it performs the following:

\begin{enumerate}[font=\bfseries,align=left, leftmargin=5pt, labelindent=5pt,listparindent=0pt, labelwidth=!, itemindent=!]
\item Compute the transmission and migration cost parameters $\beta_{c}$, $\beta_{l}$, $\mu$, $\delta_{c}$, $\delta_{l}$,
and $\theta$ based on the measurements at BSs. 

\item \label{traceOverallProc:estEmpiricalProb} Compute the average of empirical probabilities $f_{n}$ by 
\begin{align}
\bar{f}=\frac{1}{N_{\textrm{BS}}} \sum_{n\in\mathcal{N}_\textrm{BS}} f_{n}
\label{eq:realWorldEst2}
\end{align}
where $\mathcal{N}_\textrm{BS}$ is the set of BSs and $N_{\textrm{BS}}=|\mathcal{N}_\textrm{BS}|$ is the total number of BSs.
Then, estimate the parameter $r$ by 
\begin{align} 
\hat{r}=\bar{f}/6.
\label{eq:realWorldEst3}
\end{align} 

\item In the distance-based MDP, set $p_{0}=6\hat{r}$, $p=2.5\hat{r}$, and $q=1.5\hat{r}$ (as discussed in Section~\ref{sub:approxMethodDescription}), compute and save the optimal distance-based policy from Algorithm~\ref{alg:modifiedPolicyIteration}. 
Also save the estimated cost parameters and the optimal discounted sum costs $V^*(d)$ for all $d$  for later use.
 \label{traceOverallProc:estProcedureFindPolicyStep}

\end{enumerate}

\end{enumerate}

\end{enumerate}

\emph{Remark: } 
In the procedure presented above, we have assumed that $m_n (\tau) \neq 0$ for all $n$ and $\tau$. This is only for ease of presentation. When there exist some $n$ and $\tau$ such that $m_n (\tau) = 0$, we can simply ignore those terms (set the corresponding terms to zero) in the sums of (\ref{eq:realWorldEst1}) and (\ref{eq:realWorldEst2}), and set the values of $T_w$ and $N_\textrm{BS}$ to the actual number of terms that are summed up in (\ref{eq:realWorldEst1}) and (\ref{eq:realWorldEst2}), respectively. 

In essence, the above procedure recomputes $\hat{r}$ and other model parameters at an interval of $T_u$ timeslots, using measurements obtained in the previous $T_w$ slots. This allows the MDP model and the algorithm to adapt to the most recent characteristics of the system, which may dynamically change over time due to network and user dynamics.

\subsection{Service Placement Update}
\label{sec:actualServicePlacement}

At the start of every timeslot, service placement is updated and migration is performed when needed (step~\ref{traceOverallProc:placementDecision} in Section~\ref{sub:overallProcedureTraces}).

The policy found from the MDP model specifies which cell to migrate to when the system is in a particular state $(u(t), h(t))$. However, we may not be able to apply the  policy directly, because not every BS has an ES attached to it and each ES has a capacity limit. We may need to make some modifications to the service placement specified by the policy, so that the practical constraints are not violated. The MDP model also does not specify where to place the service if it was not present in the system before. In the following, we present a method to determine the service placement with these practical considerations, which is guided by the optimal policy obtained from the MDP model and at the same time satisfies the practical constraints.

The algorithm first ignores the ES capacity limit and repeat the process in steps \ref{subsub:initialPlacementTraces} and \ref{subsub:serviceMigTraces} below for every service. Then, it incorporates the capacity limit, and reassign the locations of some services (in step \ref{subsub:serviceRelocExceedCapTraces}) that were previously assigned to an ES whose capacity is exceeded. 

\begin{enumerate}[(I), font=\bfseries, align=left, leftmargin=0pt, labelindent=10pt,listparindent=0pt, labelwidth=!, itemindent=!]

\item \textbf{Initial Service Placement:} 
\label{subsub:initialPlacementTraces}
When the service was not running in the system before (i.e., it is being initialized), the service is placed onto an
ES that has the smallest distance to the user. The intuition for this rule is that the initialization cost and the cost of
further operation is usually small with such a placement. 

\item \textbf{Dynamic Service Migration:} 
\label{subsub:serviceMigTraces}
When the service has been initialized earlier and is currently running in the system, a decision on whether and where to migrate the service is made.
Without loss of generality, assume that the current timeslot is $t=0$.
We would like to find an action $a$ that is the solution to the following problem:
\begin{align}
\min_a \,\, & C_{a}(d) + \gamma \!\!\!\!\sum_{d(1) ={a}(d)-1}^{{a}(d)+1} \!\!\!\! P{[{{a}(d),d(1)}]}\cdot V^*(d(1)) 
\label{eq:balanceEquInRealWorldForDecision} \\
\parbox{1.2em}{s.t. \\ }\,\, &\parbox{20em}{there exists an ES such that the user-service distance is  $a(d)$ after migration} \nonumber
\end{align}
where $V^*(d)$ stands for the optimal discounted sum cost found  from step \ref{traceOverallProc:estProcedureFindPolicyStep} in Section \ref{sub:overallProcedureTraces}.
We note that (\ref{eq:balanceEquInRealWorldForDecision}) is a one-step value iteration following the balance equation  (\ref{eq:balanceDiscountedSumCost}).
Intuitively, it means that assuming the optimal actions are taken in future slots, find the action for the current slot that incurs the lowest discounted sum cost (including both immediate and future cost). When all BSs have ESs attached to them, the solution $a$ to problem (\ref{eq:balanceEquInRealWorldForDecision}) is the same as the optimal action of the MDP-model found from step \ref{traceOverallProc:estProcedureFindPolicyStep} in Section \ref{sub:overallProcedureTraces}. However, the optimal $a$ from (\ref{eq:balanceEquInRealWorldForDecision}) may be different from the optimal action from the MDP model when some BSs do not have ESs attached.
The resulting distance-based migration action can be mapped to a migration action on 2-D space using the procedure in Section~\ref{sub:approxMethodDescription}.

\item \textbf{Reassign Service Location if ES's Capacity Exceeded:} 
\label{subsub:serviceRelocExceedCapTraces}
The placement decisions in steps \ref{subsub:initialPlacementTraces} and \ref{subsub:serviceMigTraces} above do not consider the capacity limit of each ES, so it is possible that we find the ES capacity is exceeded after following the steps in the above sections. When this happens, we start with an arbitrary ES (denoted by $i_0$) whose capacity constraint is violated. We rank all services in this ES according to the objective function in (\ref{eq:balanceEquInRealWorldForDecision}), and start to reassign the service location with the highest objective function value. This service is placed on an ES that still has capacity for hosting it, where the placement decision is also made according to (\ref{eq:balanceEquInRealWorldForDecision}) but only the subset of ESs that are still capable of hosting this service are considered. 
The above process is repeated until the number of services hosted at ES $i_0$ becomes within its capacity limit. Then, this whole process is repeated for other ESs with exceeded capacity.

\end{enumerate}

\emph{Remark}: We note that service relocation does not really occur in the system. It is only an intermediate step in the algorithm for finding new service locations.
We use this two-step approach involving temporary placement and relocation instead of an alternative one-step approach that checks for ES capacity when performing the placement/migration in steps \ref{subsub:initialPlacementTraces} and \ref{subsub:serviceMigTraces}, because with such a two-step approach, we can leave the low-cost services within the ES and move high-cost services to an alternative location.

\subsection{Performance Analysis}

\subsubsection{Estimation of Parameter $r$} 
\label{section:EstDiscussion}

As introduced in Section \ref{section:2D1DApprox}, at every timeslot, each user randomly moves to one of its neighboring cells with probability $r$ and stays in the same cell with probability $1-6r$. In the real world, the parameter $r$ is unknown a priori and needs to be estimated based on observations of user movement. Equations (\ref{eq:realWorldEst1})--(\ref{eq:realWorldEst3}) in Section~\ref{sub:overallProcedureTraces} serve for this estimation purpose, and the resulting $\hat{r}$ is an estimator of $r$. In the following, we analyze some statistical properties of $\hat{r}$ and discuss the rationale for using such an estimation approach\footnote{In the analysis, we still assume that $m_n (\tau) \neq 0$ for all $n$ and $\tau$ as in Section~\ref{sub:overallProcedureTraces}. For cases with $m_n (\tau) = 0$, we can make a similar substitution as described in the remark in Section~\ref{sub:overallProcedureTraces}.}.

We note that the mobility model presented in Section~\ref{section:2D1DApprox} is for an infinite 2-D space with an infinite number of cells. In reality, the number of cells is finite. We assume in our analysis that each user stays in the same cell with probability $1-6r$ ($r\leq \frac{1}{6}$) and moves out of its current cell with probability $6r$, no matter whether the cell is at the boundary (such as cells in the outer ring $i=3$ in Fig.~\ref{fig:hexCell}) or not. 
When a cell is at the boundary, its transition probability to each of its neighboring cells is larger than $r$, because it has less than six neighbors. For example, in  Fig.~\ref{fig:hexCell}, a user in cell $(3,0)$ moves to \emph{each} of its neighboring cells (including $(3,17), (2,0), (3,1)$) with probability $2r$, and the total probability of moving out of the cell is still $6r$.
We also assume that the mobility patterns of different users are independent of each other.

\begin{theorem}
\label{prop:rEstUnbiased}
Assume that each user follows 2-D random walk (defined above) with parameter $r$, then $\hat{r}$ is an unbiased estimator of $r$, i.e., $\expect{\hat{r}}=r$.
\end{theorem}
\begin{proof}
See Appendix \ref{sec:proofOfREstUnbiased}.
\end{proof}

The fact that $\hat{r}$ is an unbiased estimator of $r$ justifies our estimation approach, which intuitively means that the long-term average of the estimated value $\hat{r}$ should not be too far away from the true value of $r$. 

We analyze the variance of the estimator next. Such analysis is not very easy due to the dependency among different random variables.
To make the analysis theoretically tractable, we introduce the following assumption.
\begin{assumption}
\label{assumption:independence}
We assume that $m'_n(\tau)$ is independent of $m_n(\tilde{\tau})$, $m_{\tilde{n}}(\tau)$, $m_{\tilde{n}}(\tilde{\tau})$, $m'_n(\tilde{\tau})$, $m'_{\tilde{n}}(\tau)$, and $m'_{\tilde{n}}(\tilde{\tau})$ (where $\tilde{n} \neq n$ and $\tilde{\tau} \neq \tau$) when $m_n(\tau)$ is given.
\end{assumption}
Essentially, this assumption says that $m'_n(\tau)$ is only dependent on $m_n(\tau)$. In practice when the number of users is large, this assumption is close to reality, because the locations of different users are independent of each other, and also because of the Markovian property which says that future locations of a user only depends on its present location (and independent of all past locations when the present location is given).

\begin{theorem}
\label{prop:rEstVarianceUpperBound}
Assume that Assumption \ref{assumption:independence} is satisfied and each user follows {2-D} random walk (defined at the beginning of Section \ref{section:EstDiscussion}) with parameter $r$.
The variance of estimator $\hat{r}$ is upper bounded by $\mathrm{Var} \{ \hat{r} \} \leq \frac{1}{144 N_{\textrm{BS}} T_w}$.
\end{theorem}
\begin{proof}
See Appendix \ref{sec:proofOfREstVarianceUpperBound}.
\end{proof}

We see from Theorem \ref{prop:rEstVarianceUpperBound} that the upper bound of variance is inversely proportional to $T_w$. This is an intuitive and also very favorable property, which says that when users follow an ideal random walk mobility model with parameter $r$, we can estimate the value of $r$ as accurate as possible if we have a sufficient amount of samples. 

Different from many estimation problems where it is costly (requiring  human participation, data communication, etc.) to obtain samples, it is not too difficult to collect samples in our case, because each BS can save user records at every timeslot, and we can easily adjust the number of samples (proportional to $T_w$) by changing the amount of timeslots to search back in the record. Therefore, unlike many other estimation problems where the goal is to minimize the variance of estimator under a given sample size, we do not target this goal here. A much more important issue in our problem is that the ideal random walk mobility model may not hold in practice. The parameter estimation procedure in Section \ref{sub:overallProcedureTraces} considers possible model violations, and its rationale is explained below.

In the estimation procedure, we first focus on a particular cell $n$ and compute the empirical probability of users moving out of that cell in (\ref{eq:realWorldEst1}), by treating each timeslot with equal weight. Then, we take the average of such empirical probabilities of all cells in (\ref{eq:realWorldEst2}). We note that the number of users in different cells and timeslots may be imbalanced. Thus, in (\ref{eq:realWorldEst1}), the empirical probabilities for different cells and slots may be computed with different number of users (samples), i.e., different values of $m_n(\tau)$. 

Suppose we use an alternative approach that first computes the total number of users in all cells and all slots (i.e., $\sum_{n\in\mathcal{N}_\textrm{BS}}  \sum_{ \tau=t-T_w}^{t-1} m_{n}(\tau)$) and the aggregated total number of users leaving their current cells at the end of slots (i.e., $\sum_{n\in\mathcal{N}_\textrm{BS}} \sum_{ \tau=t-T_w}^{t-1} m'_{n}(\tau)$). Then, this approach estimates $r$ by the overall empirical probability of users moving out of their current cells (i.e., $ \frac{\sum_{n\in\mathcal{N}_\textrm{BS}} \sum_{ \tau=t-T_w}^{t-1} m'_{n}(\tau)}{6 \cdot \sum_{n\in\mathcal{N}_\textrm{BS}} \sum_{ \tau=t-T_w}^{t-1} m_{n}(\tau)} $).
Intuitively, this alternative estimator may bring a lower variance compared to our estimator $\hat{r}$, because it treats all the samples as a whole. 

However, the uniform random walk mobility model \emph{may not hold} precisely in practice; even if it holds, the parameter $r$ \emph{may be time-varying or different in different geographical areas}. Thus, we compute the empirical probability for each cell in each slot first, so that different cells and slots are equally weighted in the estimation, although some cells may have more users than others in certain slots. This is for the consideration of fairness among different cells, so that the performance at different cells remains similar. It is also to avoid statistics at a single slot dominating the overall result. Since $r$ may be time-varying in practice, it is possible that single-slot statistics do not represent the overall case.

\subsubsection{Cost Difference Due to Additional Constraints}
Due to the constraints on ES existence and capacity, the procedure in Section~\ref{sec:actualServicePlacement} may have to choose an action $a(d)$ that is different from the optimal action $a^*(d)$ of the distance-based MDP. The following theorem captures the effect of such non-optimal choices, from the MDP's point of view.

\begin{theorem}\label{prop:approxErrorBoundAddConstraints} Assume that the actual action $a(d) = \left[a^*(d) + k\right]_0^{N-1}$ with probability $\alpha_k$, and there exists $K\geq 0$ such that $\alpha_k = 0$ for all $|k|>K$. Let $\tilde{V}(d)$ and $V^*(d)$ denote the discounted sum costs of the distance-based MDP, when following the actual and optimal actions, respectively. We have  $\tilde{V}(d)-V^*(d)\leq \frac{\gamma \kappa'}{1-\gamma}$ for all $d$, where $\kappa' \triangleq  \max_x \left\{b\left(x + K+2\right) - b(x) \right\}$ and $[x]_v^w$ denotes $\min\{\max\{x,v\},w\}$.
\end{theorem}
\begin{proof}
See Appendix~\ref{sec:ProofOfApproxErrorBoundAddConstraints}.
\end{proof}

In the above theorem, we assume that the actual actions are random for the ease of analysis. 
The probabilities $\alpha_k$ (and $K$) depend on the additional constraints on ES existence/capacity and the load of the system. Intuitively, $K$ becomes large when the system is heavily loaded or when only a few of the BSs have ESs attached to them, because the system may need to choose an action that is far away from the optimal action in this case. Since the migration cost $b(x)$ is non-decreasing with $x$, a larger $K$ can only give a worse (or equal, but not better) error bound, as one would intuitively expect.

\subsection{Trace-Driven Simulation}

\label{sub:realWorldSimulation}

\begin{figure}
\center{
\subfigure[]{\includegraphics[width=0.4\columnwidth]{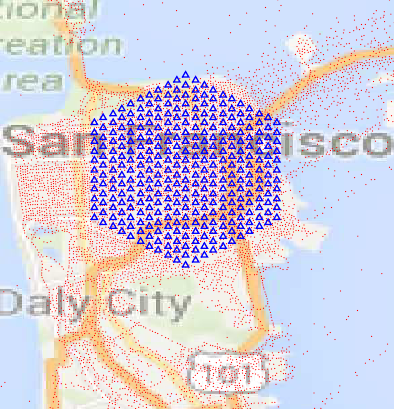}}
\quad
\subfigure[]{\includegraphics[width=0.4\columnwidth]{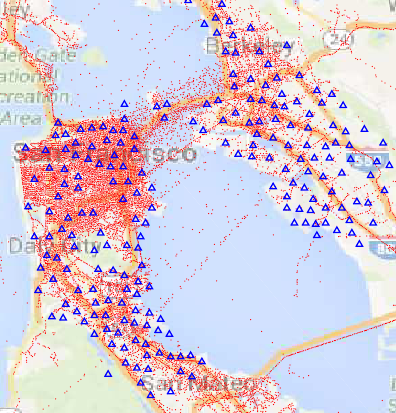}}
}

\protect\caption{BS and taxi locations: (a) Hexagonal BS placement, (b) Real BS placement. The blue triangles indicate the BS location and the red dots indicate possible taxi locations. There appears to be a small amount of erroneous/inaccurate data for taxi locations but the majority of them are correct. Map data: Google.}
\label{fig:simCellsMap} 
\end{figure}

\begin{figure*}
\center{

\begin{tabular}{ c c c c}
 \subfigure{\small \raisebox{2cm}{$\mathcal{H}$+$\mathcal{N}$:}} & 
 \subfigure{\includegraphics[width=0.3\linewidth]{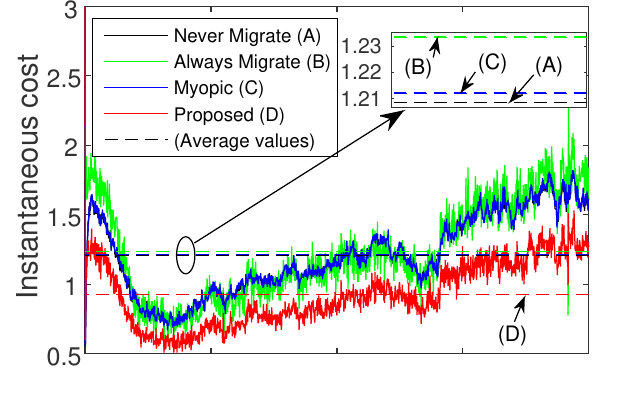}} & 
\subfigure{\includegraphics[width=0.29\linewidth]{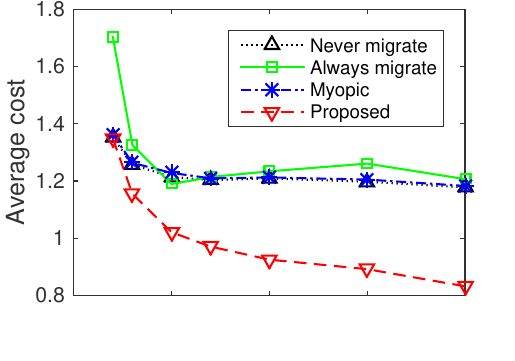}} &
\subfigure{\includegraphics[width=0.29\linewidth]{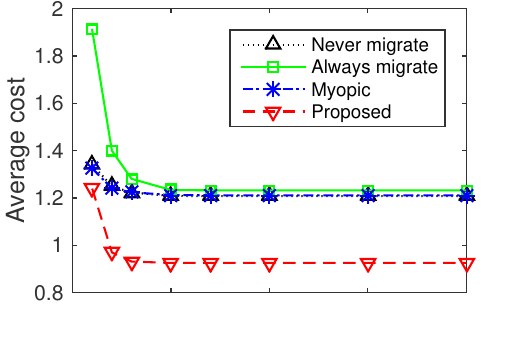}} 
\vspace{-0.2in} \\ 

\subfigure{\small  \raisebox{2cm}{$\mathcal{H}$+$\mathcal{C}$:}} &
\subfigure{\includegraphics[width=0.3\linewidth]{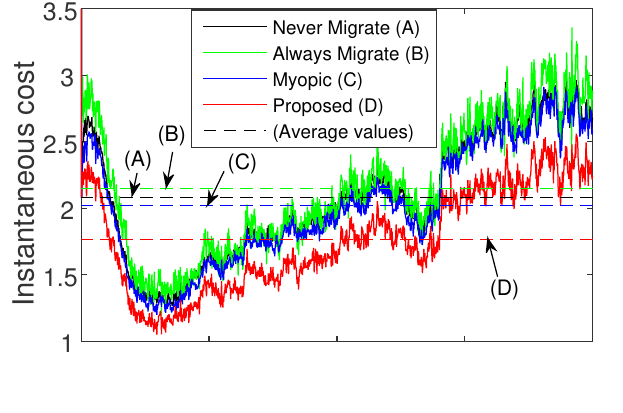}} &
\subfigure{\includegraphics[width=0.29\linewidth]{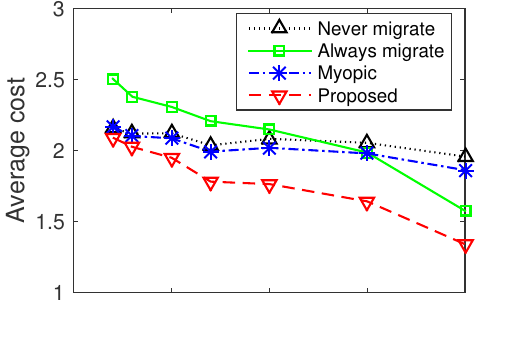}} &
\subfigure{\includegraphics[width=0.29\linewidth]{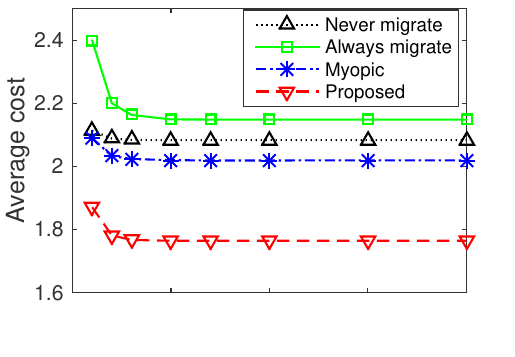}}
\vspace{-0.2in} \\ 

\subfigure{\small  \raisebox{2cm}{$\mathcal{R}$+$\mathcal{N}$:} } &
\subfigure{\includegraphics[width=0.3\linewidth]{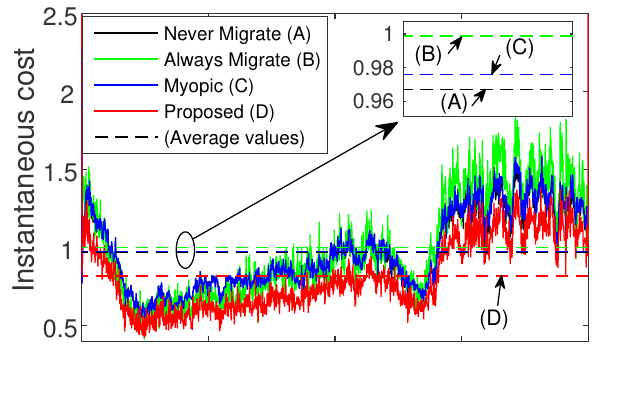}} &
\subfigure{\includegraphics[width=0.29\linewidth]{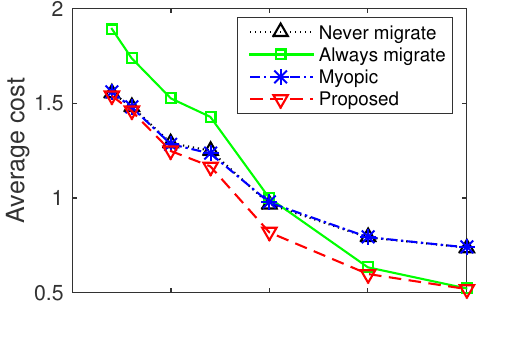}} &
\subfigure{\includegraphics[width=0.29\linewidth]{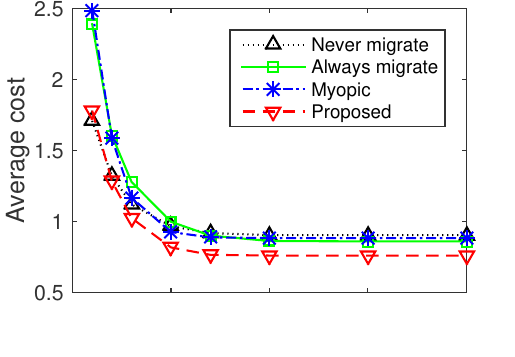}}
\vspace{-0.2in} \\ 

\subfigure{\small  \raisebox{2cm}{$\mathcal{R}$+$\mathcal{C}$:} } &
\subfigure{\includegraphics[width=0.3\linewidth]{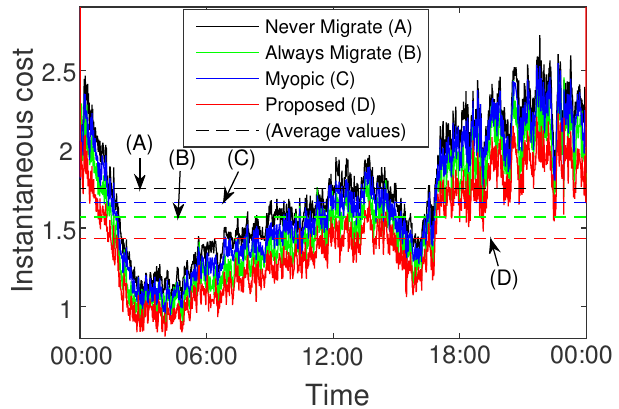}} &
\subfigure{\includegraphics[width=0.29\linewidth]{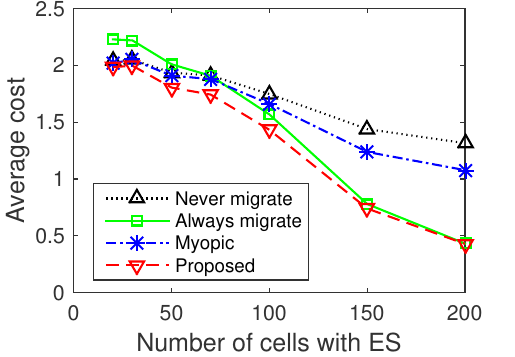}} &
\subfigure{\includegraphics[width=0.29\linewidth]{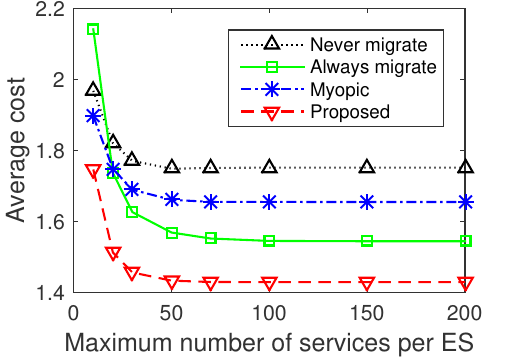}}
\vspace{-0.05in}
\end{tabular}
}
\protect\caption{Instantaneous and average costs over a day in trace-driven simulation, where $R_t=R_p=1.5$, $\mathcal{H}$ and $\mathcal{R}$ respectively stand for $\mathcal{H}$exagonal and $\mathcal{R}$eal BS placement, $\mathcal{N}$ and $\mathcal{C}$ respectively stand for $\mathcal{N}$on-constant and $\mathcal{C}$onstant costs. In some plots, an enlarged plot of the circled area is shown. The arrows annotated with (A), (B), (C), and (D) point to the average values over the whole day of the corresponding policy.}
\label{fig:simInstant} 
\vspace{-0.1in}
\end{figure*}

We perform simulation with real-world mobility traces of $536$ taxis in San Francisco, collected on  May 31, 2008 \cite{comsnets09piorkowskiMobilityTraces,epfl-mobility-2009-02-24}, where different number of taxis are operating (i.e., active) at different time of the day. Each active taxi is modeled as a user that requests a service that is independent of the services of the other taxis.  Two different ways of BS placement are considered, as shown in Fig.~\ref{fig:simCellsMap}. The first is a \emph{hexagonal} BS placement with $500$~m cell separation and $N_\textrm{BS}=331$ cells (BSs) in total in the central area of San Francisco. The second placement uses \emph{real} cell tower locations obtained from \url{www.antennasearch.com} on a larger geographical area. To limit the total number of BSs in the simulation, the cell tower locations in the original dataset are down-selected so that the distance between two BSs is at least $1,000$~m.
In both cases, each taxi connects to its nearest BS measured by Euclidean distance. 
This modeling of user mobility and BS locations is similar to other recent work including \cite{ZhangCaching2015,Saurez2016Migration,WangICDCS2017,zhang2017VR}.
For the hexagonal BS placement, the distance of the MDP model is computed as the length of the shortest path between two different BSs. For the real BS placement, the distance is computed as the Euclidean distance between BSs divided by $1,000$~m and rounded up to the next integer.

We set the physical time length for each timeslot as $60$~s. The parameters for data collection and policy update are set as $T_u=1$~slot, and $T_w=60$~slots. We choose $N=10$ and $\gamma=0.9$. 
From Theorem \ref{prop:rEstVarianceUpperBound}, we know that for the hexagonal BS placement, the standard deviation of estimator $\hat{r}$ (for an ideal random walk mobility model) is upper-bounded by $\sqrt{\frac{1}{144 N_\textrm{BS} T_w}}=0.00059$, which is reasonably small.
Unless otherwise specified, there are $100$ BSs that have ESs attached to them and each ES can host at most $50$ services. The BSs with ESs are evenly distributed among all BSs. Note that the locations of taxis in the dataset are unevenly distributed and the density of taxis can be very high in some areas at certain times of the day.

Similar to Section \ref{numEval2D1DApprox}, we compare the performance of the proposed method with baseline policies including always/never-migrate and myopic. To cope with the nature of the policies, the objective function in (\ref{eq:balanceEquInRealWorldForDecision}) is modified accordingly for these baseline policies. The objective in (\ref{eq:balanceEquInRealWorldForDecision}) is defined as the user-service distance for the always- or never-migrate policies (recall that migration also happens with never-migrate policy when the user-service distance is $N$ or larger), and it is defined as the one-timeslot cost for the myopic policy.

\textbf{Cost definition:} The cost is defined as related to both the distance and the system load, where the latter may introduce additional delay due to queueing. We assume that the system load is proportional to the number of taxis in operation, and define parameters $R_t, R_p > 1$ representing the availability of two types of resources. Then, we define  $G_t \triangleq 1\Big/\left(1-\frac{m_{\textrm{cur}}}{R_t m_{\textrm{max}}}\right)$ and $G_p \triangleq 1\Big/\left(1-\frac{m_{\textrm{cur}}}{R_p m_{\textrm{max}}}\right)$,
where $m_{\textrm{cur}}$ denotes the number of taxis in operation at the time when the optimal policy is being computed, and $m_{\textrm{max}}$ denotes the maximum number of taxis that may simultaneously operate at any time instant in the considered dataset. 
The above expressions for $G_t$ and $G_p$ have the same form as the average queueing delay expression \cite{refQueuingBook,LiLoadBalancing}.

We set $\mu=\theta=0.8$ and consider two different cost definitions. The first is \emph{non-constant cost} with parameters defined as $\beta_c=G_p+G_t$, $\beta_l=-G_t$, $\delta_c=G_t$, and $\delta_l=-G_t$. 
With such a definition, according to (\ref{eq:costMigration}) and (\ref{eq:costTransmission}), we have that whenever the distance $x,y>0$, the migration cost ${b}(x) >\beta_c+\beta_l = G_p$ representing a non-zero processing cost whenever the distance is not zero, and the transmission cost ${c}(y) > \delta_c+\delta_l = 0$. The costs ${b}(x)$ and ${c}(y)$ increase with $x$ and $y$ respectively, at a rate related to $G_t$, to represent the cost due to network communication (incurred in both migration and user-ES data transmission) which is related to the distance. 
The second type of cost we consider is \emph{constant cost}, where we set $\beta_c=G_p$, $\delta_c=G_t$, and $\beta_l = \delta_l = 0$. With this definition, the migration cost is $G_p$ and the transmission cost is $G_t$ whenever the distance is larger than zero.

\textbf{Results:} The results are shown in Fig.~\ref{fig:simInstant}, which includes the instantaneous costs over the day (i.e., the $C_a(s(t))$ values, averaged over all users that are active in that slot), and the average costs over the day with different number of cells with ES and different capacities per ES. We see that proposed approach gives lower costs than other approaches in almost all cases.
The fluctuation in the instantaneous cost is due to different system load (number of active taxis) over the day.
This also implies that it is necessary to compute the optimal policy in real-time, based on recent observations of the system condition.
The average cost of the proposed approach becomes closer to that of never-migrate and myopic policies when the number of ES or the capacity per ES is small, because it is hardly possible to migrate the service to a better location due to the constrained resource availability in this case.
Additional simulation results are presented in Appendix~\ref{sec:AdditionalSimulationResults}, showing that the proposed approach outperforms the other approaches also with different values of $R_t, R_p$.

\section{Discussions}
\label{section:Discussions}

Some assumptions have been made in the paper to make the problem theoretically tractable and also for the ease of presentation. In this section, we justify these assumptions from a practical point of view and discuss possible extensions.

\textbf{Cost Functions:} To ease our discussion, we have limited our attention to transmission and migration costs in this paper. This can be extended to include more sophisticated cost models. For example, the transmission cost can be extended to include the computational cost of hosting the service at an ES, by adding a constant value to the transmission cost expression.
As in Section~\ref{sub:realWorldSimulation}, the cost values can also be time-varying and related to the background system load. 
Furthermore, the cost definition can be practically regarded as the average cost over multiple locations, which means that when seen from a single location, the monotonicity of cost values with distances does not need to apply. This makes the proposed approach less restrictive in terms of practical applicability.

We also note that it is generally possible to formulate an MDP with  additional dimensions in cost modeling, such as one that includes the state of the network, load at each specific ES, state of the service to avoid service interruption when in critical state, etc. However, this requires a significantly larger state space compared to our formulation in this paper, as we need to include those network/ES/service states in the state space of the MDP. There is a tradeoff between the complexity of solving the problem and accuracy of cost modeling. Such issues can be studied in the future, where we envision similar approximation techniques as in this paper can be used to approximately solve the resulting MDP.

\textbf{Single/Multiple Users:} As pointed out in Section \ref{section:problem_formulation}, although we have focused on a single user in our problem modeling, practical cases involving multiple users running independent services can be considered by setting cost functions related to the background traffic generated by other users, as in Section \ref{section:RealWorldTraces}. For more complicated cases such as multiple users sharing the same service, or where the placement of different services is strongly coupled and reflected in the cost value, we can formulate the problem as an MDP with larger state space (similarly to the generalized cost model discussed above).
 
\textbf{Uniform Random Walk:}
The uniform random walk mobility model is used as a modeling assumption, which not only simplifies the theoretical analysis, but also makes the practical implementation of the proposed method fairly simple in the sense that only the empirical probability of users moving outside of the cell needs to be recorded (see Section \ref{sub:overallProcedureTraces}). This model can capture the average mobility of a large number of users. The simulation results in Section \ref{sub:realWorldSimulation} confirm that this model provides good performance, even though individual users do not necessarily follow a uniform random walk.

\textbf{Centralized/Distributed Control:}
We have focused on a centralized control mechanism in this paper for the ease of presentation. However, many parts of the proposed approach can be performed in a distributed manner. For example, in step \ref{traceOverallProc:placementDecision} in Section \ref{sub:overallProcedureTraces}, the service placement decision can be made among a smaller group of ESs if the controller sends the results from step \ref{traceOverallProc:estProcedureFindPolicyStep} to these ESs. In particular, the temporary service placement decision in steps \ref{subsub:initialPlacementTraces} and \ref{subsub:serviceMigTraces} in Section~\ref{sec:actualServicePlacement} can be made locally on each ES (provided that it knows the locations of other ESs). The capacity violation check in step \ref{subsub:serviceRelocExceedCapTraces} in Section~\ref{sec:actualServicePlacement} can also be performed locally. If some ESs are in excess of capacity, service relocation can be performed using a few control message exchange between ESs, where the number of necessary messages is proportional to the number of services to relocate. Relocation would rarely occur if the system load is not very high. The computation of average empirical probability in step \ref{traceOverallProc:estEmpiricalProb} in Section \ref{sub:overallProcedureTraces} can also be distributed in the sense that a subset of ESs compute local averages, which are subsequently sent to the MEC controller that computes the global average.

\textbf{ES-BS Co-location:}
For ease of presentation, we have assumed that ESs are co-located with BSs. However, our proposed approach is not restricted to such cases and can easily incorporate scenarios where ESs are not co-located with BSs as long as the costs are geographically dependent.

\section{Conclusions}
\label{section:Conclusions}

In this paper, we have studied service migration in MEC. The problem is formulated as an MDP, but its state space can be arbitrarily large. To make the problem tractable, we have reduced the general problem into an MDP that only considers a meaningful parameter, namely the distance between the user and service locations. The  distance-based MDP has several structural properties that allow us to develop an efficient algorithm to find its optimal policy. We have then shown that the distance-based MDP is a good approximation to scenarios where the users move in a \mbox{2-D} space, which is confirmed by analytical and numerical evaluations. After that, we have presented a method of applying the theoretical results to a practical setting, which is evaluated by simulations with real-world data traces of taxis and cell towers in San Francisco.

The results in this paper provide an efficient solution to service migration in MEC. Further, we envision that our theoretical approach can be extended to a range of other MDP problems that share similar properties. The highlights of our approaches include: a closed-form solution to the discounted sum cost of a particular class MDPs, which can be used to simplify the procedure of finding the optimal policy; a method to approximate an MDP (in a particular class) with one that has smaller state space, where the approximation error can be shown analytically; and a method to collect statistics from the real-world to serve as parameters of the MDP.

\section*{Acknowledgment}
Contribution of S. Wang is partly related to his previous affiliation with Imperial College London. Contributions of R.~Urgaonkar, M. Zafer, and T. He are related to their previous affiliation with IBM Research.

\bibliographystyle{IEEEtran}
\bibliography{IEEEabrv,main}

\clearpage 

\appendices

\setcounter{theorem}{0}
\setcounter{corollary}{0}
\setcounter{equation}{0}
\setcounter{figure}{0}
\renewcommand{\thetheorem}{A.\arabic{theorem}}
\renewcommand{\thecorollary}{A.\arabic{corollary}}
\renewcommand{\theequation}{A.\arabic{equation}}
\renewcommand{\thefigure}{A.\arabic{figure}}

\section{Proof of Theorem \ref{theorem:notMigrateToFurther}}
\label{sec:proofOfTheoremNotMigrateFurther}

Suppose that we are given a service migration policy $\pi$ such that the service can be migrated to a location that is farther away from the user, i.e., $\Vert u-h'\Vert > \Vert u-h\Vert $. We will show that, for an arbitrary sample path of the user locations $u(t)$, we can find a (possibly history dependent) policy $\psi$ that does not migrate to locations farther away from the user in any timeslot, which performs not worse than policy~$\pi$. 

For an arbitrary sample path of user locations $u(t)$, denote the timeslot $t_{0}$ as the \emph{first} timeslot (starting from $t=0$) in which the service is migrated to somewhere farther away from the user when following policy $\pi$. The initial state at timeslot $t_0$ is denoted by $s(t_0)=(u(t_0),h(t_0))$. When following $\pi$, the state shifts from $s(t_{0})$ to $s'_{\pi}(t_{0})=(u(t_0),h'_\pi(t_0))$, where $\Vert u(t_0)-h'_\pi(t_0)\Vert >\Vert u(t_0)-h(t_0) \Vert $. The subsequent states in this case are denoted by $s_{\pi}(t)=(u(t),h_\pi(t))$ for $t> t_0$.

Now, we define a policy $\psi$ such that the following conditions are satisfied for the given sample path of user locations $u(t)$:
\begin{itemize}
\item The migration actions in timeslots $t<t_0$ are the same when following either $\psi$ or $\pi$.
\item The policy $\psi$ specifies that there is no migration within the timeslots $t\in[t_0,t_m-1]$, where $t_m>t_0$ is a timeslot index that is defined later.
\item The policy $\psi$ is defined such that, at timeslot $t_m$, the service is migrated to $h'_\pi(t_m)$, where $h'_\pi(t_m)$ is the service location (after possible migration at timeslot $t_m$) when following $\pi$.
\item For timeslots $t>t_m$, the migration actions for policies $\psi$ and $\pi$ are the same.
\end{itemize}
For $t> t_0$, the states when following $\psi$ are denoted by $s_{\psi}(t)=(u(t),h_\psi(t))$.

The timeslot $t_m$ is defined as the \emph{first} timeslot after $t_0$ such that the following condition is satisfied:
\begin{enumerate}
\item $\Vert u(t_m)-h_\psi(t_m) \Vert > \Vert u(t_m)-h'_\pi(t_m) \Vert$, i.e., the transmission cost (before migration at timeslot $t_m$) when following $\psi$ is larger than the transmission cost (after possible migration at timeslot $t_m$) when following $\pi$.
\end{enumerate}

Accordingly, within the interval $[t_0+1, t_m-1]$, the transmission cost when following $\psi$ is always less than or equal to the transmission cost when following $\pi$, and there is no migration when following $\psi$. Therefore, for timeslots $t\in [t_0, t_m-1]$, policy $\pi$ cannot bring lower cost than policy~$\psi$. 

In timeslot $t_{m}$, we can always choose a migration action for $\psi$ where the migration cost is smaller than or equal to the sum of the migration costs of $\pi$ within $[t_0,t_m]$. The reason is that $\psi$ can migrate (in timeslot $t_{m}$) following the same migration path as $\pi$ within $[t_0,t_m]$.

It follows that, for timeslots within $[t_0,t_m]$, policy $\pi$ cannot perform better than policy $\psi$, and both policies have the same costs for timeslots within $[0,t_0-1]$ and $[t_m+1,\infty)$. 
The above procedure can be repeated so that all the migration actions to a location farther away from the user can be removed without increasing the overall cost, thus we can redefine $\psi$ to be a policy that removes all such migrations.

We note that the policy $\psi$ can be constructed based on policy $\pi$ without prior knowledge of the user's future locations. For any policy $\pi$, the policy $\psi$ is a policy that does not migrate whenever $\pi$ migrates to a location farther away from the user (corresponding to timeslot $t_0$). Then, it migrates to $h'_\pi(t_m)$ when condition 1 is satisfied (this is the timeslot $t_m$ in the above discussion). 

The policy $\psi$ is a history dependent policy, because its actions depend on the past actions of the underlying policy $\pi$. From \cite[Chapter 6]{puterman2009markov}, we know that history dependent policies cannot outperform Markovian policies for our problem, assuming that the action spaces of both policies are identical for every possible state.  
Therefore, there exists a Markovian policy that does not migrate to a location farther away from the user, which does not perform worse than $\pi$. 
Noting that the optimal policy found from (\ref{eq:objFunc}) and (\ref{eq:bellman}) are Markovian policies, we have proved the theorem.

\section{Proof of Corollary \ref{corollary:constant_cost}}
\label{sec:proofOfCorollaryConstantCost}

The proof follows the same procedure as the proof of Theorem \ref{theorem:notMigrateToFurther}. For any given policy $\pi$ that migrates to locations other than the user location, we show that there exists a policy $\psi$ that does not perform such migration, which performs not worse than the original policy $\pi$. The difference is that $t_0$ is defined as the first timeslot such that $u(t_0)\neq h'_\pi(t_0)$, and $t_m$ is defined as the first timeslot after $t_0$ such that $u(t_m)=h'_\pi(t_m)$. Due to the strict inequality relationship of the cost functions given in the corollary, and because $0<\gamma<1$, we can conclude that $\pi$ is not optimal.

\section{Approximating a General Cost Function with Constant-Plus-Exponential Cost Function}
\label{supSec:approxCostFunc}

We only focus on the migration cost function ${b}(x)$, because ${b}(x)$ and ${c}(y)$ have the same form. For a given cost function $f(x)$, where $x\geq 0$, we would like to approximate it
with the constant-plus-exponential cost function given by ${b}(x)=\beta_{c}+\beta_{l}\mu^{x}$. Note that although we force ${b}(0)=0$ by definition, we relax that restriction here and only consider the smooth part of the cost function for simplicity. We assume that this smoothness can be extended to $x=0$, so that $f(x)$ remains smooth at $x=0$. Under this definition, $f(x)$ may be non-zero at $x=0$. 

Because ${b}(x)$ includes both a constant term and an exponential
term, we cannot obtain an exact analytical expression to minimize
a commonly used error function, such as the mean squared error. We can, however, solve for the parameters that minimize the error
cost function numerically. To reduce the computational complexity
compared to a fully numerical solution, we propose an approximate solution in this section.

\begin{figure}
\center{\subfigure[]{\includegraphics[width=0.8\linewidth]{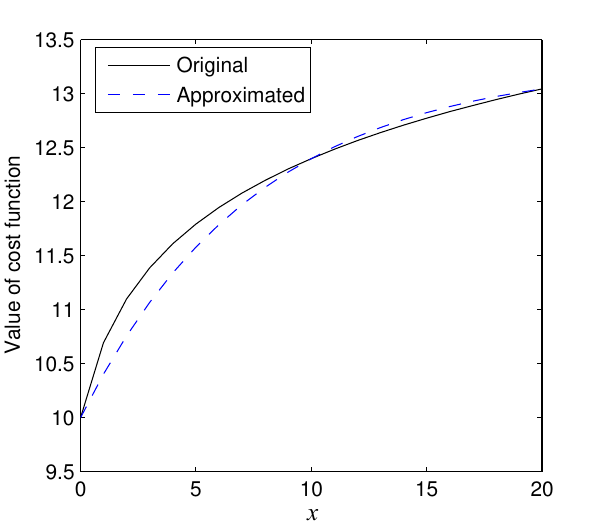}} 

\subfigure[]{\includegraphics[width=0.8\linewidth]{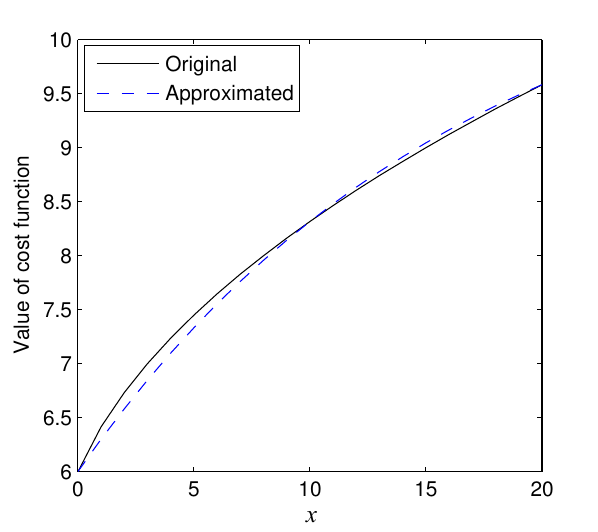}} 

\subfigure[]{\includegraphics[width=0.8\linewidth]{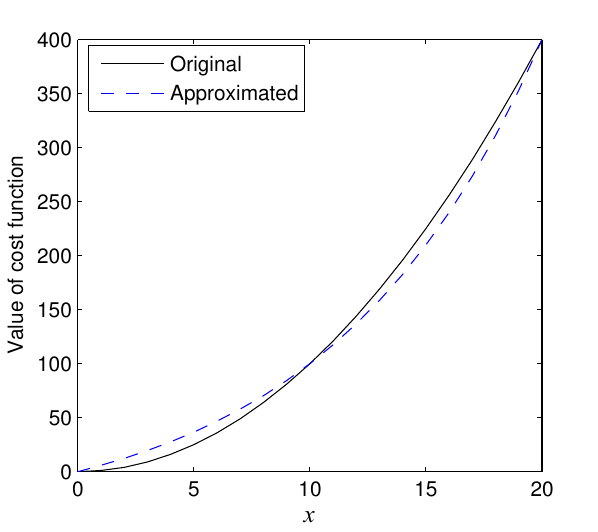}}}

\protect\caption{Examples of approximating a general cost function with exponential
cost function: (a) $f(x)=\ln(x+1)+10$, (b) $f(x)=\sqrt{x+1}+5$,
(c) $f(x)=x^{2}$.}
\label{fig:exampleApprox}
\end{figure}

We are given an\emph{ }integer $w$ that may be chosen according to
practical considerations, and we solve for the parameters $\beta_{c}$,
$\beta_{l}$, and $\mu$ according to the following system of equations:
\begin{align}
\beta_{c}+\beta_{l} & =f(0)\label{eq:approxSystem1}\\
\beta_{c}+\beta_{l}\mu^{w} & =f(w)\label{eq:approxSystem2}\\
\beta_{c}+\beta_{l}\mu^{2w} & =f(2w)\label{eq:approxSystem3}
\end{align}
Subtracting (\ref{eq:approxSystem1}) respectively from (\ref{eq:approxSystem2}) and (\ref{eq:approxSystem3})
 and dividing the two results
gives
\begin{equation}
\frac{\mu^{2w}-1}{\mu^{w}-1}=\frac{f(2w)-f(0)}{f(w)-f(0)}\label{eq:approxRatioEquation}
\end{equation}
subject to $\mu^{w}\neq1$. We can then solve $\mu^{w}$ from
(\ref{eq:approxRatioEquation}) which gives
\begin{equation}
\mu^{w}=\frac{R\pm\sqrt{R^{2}-4(R-1)}}{2}
\end{equation}
where $R\triangleq\frac{f(2w)-f(0)}{f(w)-f(0)}\geq1$. It follows
that we always have $R^{2}-4(R-1)\geq0$ and $\mu^{w}\geq0$. To
guarantee that $\mu^{w}\neq1$, we set $\mu^{w}=1\pm\epsilon_0$
if $\left|\mu^{w}-1\right|<\epsilon_0$, where $\epsilon_0$ is a small
number and the sign is the same as the sign of $\mu^{w}-1$. Then,
we have
\begin{equation}
\mu=\left(\frac{R\pm\sqrt{R^{2}-4(R-1)}}{2}\right)^{\frac{1}{w}}\label{eq:approxTheta}
\end{equation}
From which we can solve
\begin{equation}
\beta_{c}=\frac{f(0)\mu^{w}-f(w)}{\mu^{w}-1}
\end{equation}
\begin{equation}
\beta_{l}=\frac{f(w)-f(0)}{\mu^{w}-1}
\end{equation}
According to (\ref{eq:approxTheta}), we have two possible values
of $\mu$, and correspondingly two possible sets of values of $\beta_{c}$
and $\beta_{l}$. To determine which is better, we compute the sum
squared error $\sum_{x=0}^{2w}\left(f(x)-\left(\beta_{c}+\beta_{l}\mu^{x}\right)\right)^{2}$
for each parameter set, and choose the parameter set that produces
the smaller sum squared error.

Some examples of approximation results are shown in Fig.~\ref{fig:exampleApprox}.

\section{ Proof of Theorem \ref{prop:diffEquSolution} }
\label{sec:proofOfDiffEquSolution}

Note that (\ref{eq:balanceDiscountedSumCost}) is a difference equation \cite{elaydi2005introductionDiffEqu}. 
Because we only migrate at states $\{n_k \}$, we have $a(d)= d$ for $d\in (n_{k-1},n_{k})$ ($\forall k$). From (\ref{eq:balanceDiscountedSumCost}), for $d \in  (n_{k-1},n_{k})$, we have
\begin{equation}
V(d)=\delta_{c}+\delta_{l}\theta^{d}+\gamma\sum_{d_1=d-1}^{d+1}P{[{d,d_1}]} \cdot V(d_1)\,.\label{eq:sumCostDiffEqu}
\end{equation}
We can rewrite (\ref{eq:sumCostDiffEqu}) as
\begin{equation}
V(d)=\phi_{1}V(d-1)+\phi_{2}V(d+1)+\phi_{3}+\phi_{4}\theta^{d}. \label{eq:sumCostDiffEquSimlified}
\end{equation}

This difference function has characteristic
roots as expressed in (\ref{eq:diffEquMConst}).
When $0<\gamma<1$, we always have $m_{1}\neq m_{2}$. Fixing an index $k$, for $d\in (n_{k-1},n_{k})$, the homogeneous
equation of (\ref{eq:sumCostDiffEquSimlified}) has general solution
\begin{equation}
V_{h}(d)=A_{k}m_{1}^{d}+B_{k}m_{2}^{d}\,.
\end{equation} 
To solve the non-homogeneous equation (\ref{eq:sumCostDiffEquSimlified}),
we try a particular solution in the form of
\begin{equation}
V_{p}(d)=\begin{cases}
D+H\cdot\theta^{d} & \textrm{if }1-\frac{\phi_{1}}{\theta}-\phi_{2}\theta\neq0\\
D+Hd\cdot\theta^{d} & \textrm{if }1-\frac{\phi_{1}}{\theta}-\phi_{2}\theta=0
\end{cases}\label{eq:particularSolution}
\end{equation}
where $D$ and $H$ are constant coefficients.
By substituting (\ref{eq:particularSolution}) into (\ref{eq:sumCostDiffEquSimlified}),
we get
(\ref{eq:diffEquDConst}) and (\ref{eq:diffEquHConst}). 

Because the expression in (\ref{eq:sumCostDiffEquSimlified}) is related to $d-1$ and $d+1$, the result also holds for the closed interval.
Therefore, $V(d)$ can be expressed as (\ref{eq:finalSolution})  for $d\in [n_{k-1},n_{k}]$ ($\forall k$).

\section{ Proof of Theorem \ref{prop:approxErrorBound} }
\label{sec:proofOfApproxErrorBound}

The proof is completed in three steps. First, we modify the states of the 2-D MDP in such a way that the aggregate transition probability from any state $(i'_0,j'_0)\neq(0,0)$ to ring $i_1=i'_0+1$ (correspondingly, $i_1=i'_0-1$) is $2.5r$ (correspondingly, $1.5r$). We assume that we use a given policy on both the original and modified \mbox{2-D} MDPs, and show a bound on the difference in the discounted sum costs for these two MDPs. In the second step, we show that the modified \mbox{2-D} MDP is equivalent to the distance-based MDP. This can be intuitively explained by the reason that the modified 2-D MDP has the same transition probabilities as the distance-based MDP when only considering the ring index $i$, and also, the one-timeslot cost $C_a(e(t))$ only depends on $\Vert e(t)-a(e(t)) \Vert$ and $\Vert a(e(t)) \Vert$, both of which can be determined from the ring indices of $e(t)$ and $a(e(t))$. The third step uses the fact that the optimal policy for the distance-based MDP cannot bring a higher discounted sum cost for the distance-based MDP (and hence the modified 2-D MDP) than any other policy. By utilizing the error bound found in the first step twice, we prove the result. 

The detailed proof is given as follows.

Recall that among the neighbors $\left(i',j'\right)$
of a cell $(i,j)$ in the 2-D offset-based MDP $\{e(t)\}$, when $i>0$, we have two (or, correspondingly, one) cells with
$i'=i-1$, and two (or, correspondingly, three) cells with $i'=i+1$. 
To ``even out'' the different number of neighboring cells, we define
a new (modified) MDP $\left\{ g(t)\right\} $ for the 2-D offset-based model, where the
states are connected as in the original 2-D MDP, but the transition
probabilities are different, as shown in Fig.~\ref{fig:mod2DMarkov}.

In the modified MDP $\left\{ g(t)\right\} $, the transition probabilities starting from state $(0,0)$ to each of its neighboring cells have the same value $r$. For all the other states $(i,j)\neq(0,0)$, they are defined as follows:
\begin{itemize}
\item The transition probability to each of its neighbors with the same ring index $i$ is $r$.
\item If state $(i,j)$ has two (correspondingly, one) neighbors in the lower ring $i-1$, then the transition probability to each of its neighbors in the lower ring is $\frac{1.5}{2}r$ (correspondingly, $1.5r$).
\item if state $(i,j)$ has two (correspondingly, three) neighbors in the higher ring $i+1$, then the transition probability to each of its neighbors in the higher ring is $\frac{2.5}{2}r$ (correspondingly, $\frac{2.5}{3}r$).
\end{itemize}

We denote the discounted sum cost from the original MDP $\{e(t)\}$ by $V(i,j)$, and denote that from the modified MDP $\left\{ g(t)\right\}$ by $U(i,j)$, where $(i,j)$ stands for the initial state in the discounted sum cost definition.

\begin{figure}
\center{\subfigure[]{\includegraphics[width=0.9\columnwidth]{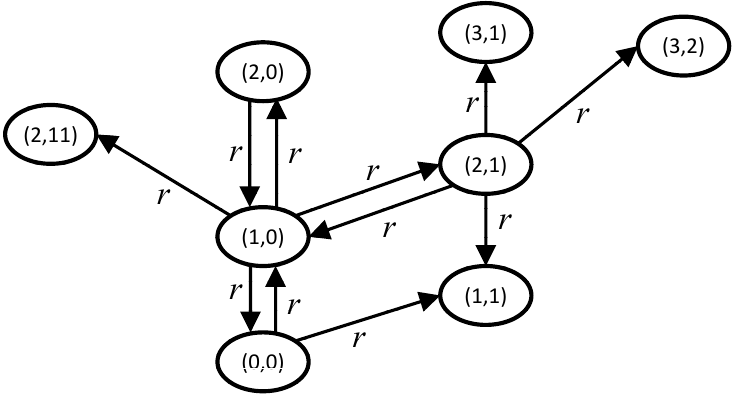}}}

\center{\subfigure[]{\includegraphics[width=0.9\columnwidth]{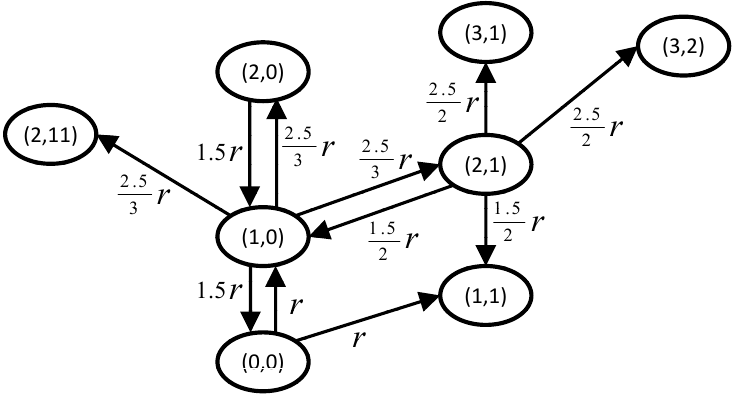}}}

\protect\caption{Illustration of original and modified 2-D MDPs, only some exemplar
states and transition probabilities are shown: (a) original, (b) modified.}
\label{fig:mod2DMarkov} 
\end{figure}

\subsection*{Part I -- Upper bound on the difference between $V(i,j)$
and $U(i,j)$ for a given policy $\pi$}

Assume we have the same policy $\pi$ for the original and modified  MDPs.
Then, the balance equation of $V(i,j)$ is
\begin{equation}
V(i,j)=C_{a}(i,j)+\gamma\!\left(\!\!\left(1-6r\right)V(a(i,j))+r\!\!\!\!\!\!\!\!\!\!\!\!\sum_{\left(i',j'\right)\in\mathcal{N}\left(a(i,j)\right)}\!\!\!\!\!\!\!\!\!\!\!\! V(i',j')\!\!\right)\label{eq:balance2DApproxW}
\end{equation}
where $a(i,j)$ is the new state after possible migration at state $(i,j)$, and $\mathcal{N}\left(a(i,j)\right)$ is the set of states that are
neighbors of state $a(i,j)$. 

For $U(i,j)$, we have
\begin{align}
& U(i,j)  =C_{a}(i,j)+\gamma\Bigg(\left(1-6r\right)U(a(i,j))\nonumber \\
 & +\!\!\!\!\!\!\sum_{\left(i',j'\right)\in\mathcal{N}\left(a(i,j)\right)}\!\!\!\!\!\! P{[{g'(t)=a(i,j), g(t+1)=(i',j')}]} \cdot U(i',j')\Bigg)\label{eq:balance2DApproxU}
\end{align}
where $P{[{g'(t)=a(i,j), g(t+1)=(i',j')}]}$ is the transition probability of the modified MDP $\{g(t)\}$ as specified earlier.

\begin{figure*}
\begin{align}
 & \left|U(i,j)-V(i,j)\right|\nonumber \\
 & =\gamma\Big|\left(1-6r\right)\left(U(a(i,j))-V(a(i,j))\right)\nonumber \\
 & \quad+\!\!\!\!\!\!\sum_{\left(i',j'\right)\in\mathcal{N}\left(a(i,j)\right)}\!\!\!\!\!\!P{[{g'(t)=a(i,j), g(t+1)=(i',j')}]} \cdot \left(U(i',j')-V(i',j')\right)\nonumber \\
 & \quad\pm0.5\lambda_{i_{a(i,j)}}r\left(\frac{\sum_{\left(i',j'\right)\in\mathcal{N}^{+}\left(a(i,j)\right)}V(i',j')}{\left|\mathcal{N}^{+}\left(a(i,j)\right) \right|}-\frac{\sum_{\left(i',j'\right)\in\mathcal{N}^{-}\left(a(i,j)\right)}V(i',j')}{\left|\mathcal{N}^{-}\left(a(i,j)\right) \right|}\right)\Big|\label{eq:2DErrorTermEqu1}\\
 & \leq\gamma\epsilon+0.5\gamma r\left|\frac{\sum_{\left(i',j'\right)\in\mathcal{N}^{+}\left(a(i,j)\right)}V(i',j')}{\left|\mathcal{N}^{+}\left(a(i,j)\right) \right|}-\frac{\sum_{\left(i',j'\right)\in\mathcal{N}^{-}\left(a(i,j)\right)}V(i',j')}{\left|\mathcal{N}^{-}\left(a(i,j)\right) \right|}\right|\label{eq:2DErrorTermEqu2}\\
 & \leq\gamma\epsilon+0.5\gamma r\max_{i,j,j':(i+1,j)\in\mathcal{N}_{2}\left(i-1,j'\right)}\left|\left(V(i+1,j)-V(i-1,j')\right)\right|\label{eq:2DErrorTermEqu3}
\end{align}

\hrulefill
\end{figure*}

In the following, let $i_{a(i,j)}$ denote the ring index of $a(i,j)$. We define sets 
\begin{align}
& \mathcal{N}^{-}(a(i,j))=\left\{ \left(i',j'\right)\in\mathcal{N}\left(a(i,j)\right):i'=i_{a(i,j)}-1\right\} \nonumber \\
& \mathcal{N}^{+}(a(i,j))=\left\{ \left(i',j'\right)\in\mathcal{N}\left(a(i,j)\right):i'=i_{a(i,j)}+1\right\} \nonumber 
\end{align}
to represent the neighboring states of $a(i,j)$ that are respectively in the lower and higher rings. We use $\left|\cdot\right|$ to denote the number of elements in a set.

Assume $\left|U(i,j)-V(i,j)\right|\leq\epsilon$ for all $i$ and
$j$, and the value of $\epsilon$ is unknown for now. We subtract
(\ref{eq:balance2DApproxW}) from (\ref{eq:balance2DApproxU}), and
then take the absolute value, yielding (\ref{eq:2DErrorTermEqu1}),
(\ref{eq:2DErrorTermEqu2}), and (\ref{eq:2DErrorTermEqu3}) which are explained below, where the set $\mathcal{N}_{2}\left(i,j\right)$ is the set of states that are two-hop neighbors of state $(i,j)$, 
the variable $\lambda_{i_{a(i,j)}}=0$ when $i_{a(i,j)}=0$, and $\lambda_{i_{a(i,j)}}=1$ when $i_{a(i,j)}>0$. 

The first two terms of (\ref{eq:2DErrorTermEqu1}) subtract the discounted
sum cost of the original MDP $\{e(t)\}$ from that of the modified MDP $\{g(t)\}$, by assuming
that both chains have the \emph{same} transition probabilities specified by the modified
MDP. The difference in their transition proabilities is captured by
the last term of (\ref{eq:2DErrorTermEqu1}). There is no difference
in the transition probabilities when $i_{a(i,j)}=0$, thus $\lambda_{i_{a(i,j)}}=0$
when $i_{a(i,j)}=0$. 

In the following, we consider $i_{a(i,j)}>0$ and further explain
the last term of (\ref{eq:2DErrorTermEqu1}). We first note that there
is difference in the transition probabilities only when moving to
the lower or higher ring:
\begin{itemize}
\item The \emph{sum} probability of moving to the lower ring in  $\left\{ g(t)\right\} $ is by $0.5r$ smaller (or, correspondingly, greater) than that in $\{e(t)\}$.
\item The \emph{sum} probability of moving to the higher ring in $\left\{ g(t)\right\} $ is by $0.5r$ greater (or, correspondingly, smaller) than that in $\{e(t)\}$.
\end{itemize}
Therefore, the transition probablity difference for each neighboring state in the lower (or higher) ring is $\pm0.5r$ divided by the number of neighbors in the lower (or higher) ring. Also note that the probablity difference for lower and higher rings have opposite signs. This explains the last term of (\ref{eq:2DErrorTermEqu1}), which captures the difference in the transition probabilities and its impact on the discounted sum costs. 

The inequality in (\ref{eq:2DErrorTermEqu2}) is from the triangle
inequality. We note that the subtraction in the last term of (\ref{eq:2DErrorTermEqu1})
only occurs on $V(i',j')$ values that are two-hop neighbors, so we
have the inequality in (\ref{eq:2DErrorTermEqu3}) by replacing the
value with the maximum.

From (\ref{eq:2DErrorTermEqu3}), we can obtain a balance equation
for the upper bound of $\left|U(i,j)-V(i,j)\right|$, which is 
\begin{equation}
\epsilon\!=\!\gamma\epsilon+0.5\gamma r\max_{i,j,j':(i+1,j)\in\mathcal{N}_{2}\left(i-1,j'\right)}\!\left|\left(V(i\!+\!1,j)\!-\! V(i\!-\!1,j')\right)\right|
\end{equation}
Because $0<\gamma<1$, the value of $V(i,j)$ converges after a number
of iterations according to (\ref{eq:balance2DApproxW}) \cite[Chapter 6]{puterman2009markov}, so $\left|\left(V(i+1,j)-V(i-1,j)\right)\right|$
also converges, and the value of $\epsilon$ can be solved by
\begin{equation}
\epsilon_{V}\!=\!\frac{\gamma r\max_{i,j,j':(i+1,j)\in\mathcal{N}_{2}\left(i-1,j'\right)}\left|\left(V(i\!+\!1,j)\!-\! V(i\!-\!1,j')\right)\right|}{2(1-\gamma)}\label{eq:epsilonExpressionW}
\end{equation}

Note that the above argument also applies when interchanging $V$
and $U$, which means that an alternative upper bound of the cost
is
\begin{equation}
\epsilon_{U}\!=\!\frac{\gamma r\max_{i,j,j':(i+1,j)\in\mathcal{N}_{2}\left(i-1,j'\right)}\left|\left(U(i\!+\!1,j)\!-\! U(i\!-\!1,j')\right)\right|}{2(1-\gamma)}\label{eq:epsilonExpressionU}
\end{equation}

The upper bound can also be expressed as $\epsilon=\min\left\{ \epsilon_{V},\epsilon_{U}\right\} $,
but either $\epsilon_{V}$ or $\epsilon_{U}$ may have the smaller
value.

\subsection*{Part II -- Optimal policy for the modified 2-D MDP $\left\{ g(t)\right\} $ is equivalent to the optimal policy for the distance-based MDP $\left\{ d(t)\right\} $}

We note that the optimal policy of an MDP can be found from value
iteration \cite[Chapter 6]{puterman2009markov}. For the modified
MDP $\left\{ g(t)\right\} $, we initialize with a never migrate
policy, which gives the initial value function $U_{0}(i,j)={c}\left(i\right)$,
satisfying $U_{0}(i,j)=U_{0}(i,j')$ for all $i$ and $j\neq j'$. 

Suppose $U_{n}(i,j)=U_{n}(i,j')$ for all $i$ and $j\neq j'$. In
each iteration, we use the following equation to obtain the new value
function and the corresponding actions for each $i$ and $j$: 
\begin{align}
 & U_{n+1}(i,j)=\min_{a}\Bigg\{ C_{a}(i,j)\nonumber \\
 & +\gamma\sum_{i'}\sum_{j'}P{[{g'(t)=a(i,j), g(t+1)=(i',j')}]} \cdot U_{n}(i',j')\Bigg\}\label{eq:approxRatioProofValueIteration}
\end{align}

From the hexagon model in Fig. \ref{fig:hexCell}, we can see that
for ring indices $i$ and $i'$, where $i'<i$, we can always reach
from state $(i,j)$ to a state in ring $i'$ with $i-i'$ hops, regardless
of the index $j$. In the $(n+1)$th iteration, if it is optimal to
migrate from state $(i,j)$ to a state with ring index $i'$, then
the migration destination must be $i-i'$ hops away from origin state,
because $U_{n}(i',j)=U_{n}(i',j')$ for all $i'$ and $j\neq j'$, 
it cannot be beneficial to migrate to somewhere farther away. 

Further, if it is optimal to migrate at a state $(i,j)$ to a state in ring $i'$, it must be optimal to migrate at states $(i,j)$ for all $j$ to a state (which may not be the same state) in ring $i'$, bringing the same cost, i.e. $U_{n+1}(i,j)=U_{n+1}(i,j')$ for $j\neq j'$. This
is due to the symmetry of cost functions (in the sense that $U_{n}(i,j)=U_{n}(i,j')$ for all $i$ and $j\neq j'$) and symmetry of transition probabilities (in the sense that the sum probability of reaching ring $i-1$ from any state in ring $i$ is the same, and the sum probability of reaching ring $i+1$ from any state in ring $i$ is also the same). Similarly, if it is optimal not to migrate at a state $(i,j)$, then it is optimal not to migrate at states $(i,j)$ for all $j$, which also brings $U_{n+1}(i,j)=U_{n+1}(i,j')$ for any $j\neq j'$.

Because the value iteration converges to the optimal policy and its corresponding cost as $n\rightarrow\infty$, for the optimal policy of the modified MDP $\left\{ g(t)\right\} $, we have the same discounted
sum cost for states in the same ring, i.e. $U^{*}(i,j)=U^{*}(i,j')$
for all $i$ and $j\neq j'$. Meanwhile, for a given $i$, the optimal actions $a^{*}(i,j)$ and $a^{*}(i,j')$ for any $j\neq j'$ have the same ring index $i_{a^{*}(i,j)}$.

Since the optimal actions $a^{*}(i,j)$ only depend on the ring index
$i$ and the ring index of $a^{*}(i,j)$ does not change with $j$,
the optimal policy for $\{g(t)\}$ can be directly mapped to a policy
for the distance-based MDP $\{d(t)\}$. A policy for $\{d(t)\}$ can also be mapped to a policy for $\{g(t)\}$ by considering the shortest path between different states in $\{g(t)\}$, as discussed in Section \ref{sub:approxMethodDescription}. This implies that there is a one-to-one mapping between the optimal policy for $\{g(t)\}$ and a policy for $\{d(t)\}$, because the optimal policy for $\{g(t)\}$ also only migrates along the shortest path between states. 

We now show that the policy for $\{d(t)\}$ obtained from the optimal policy for $\{g(t)\}$ is optimal for $\{d(t)\}$. To find the optimal policy for $\{d(t)\}$, we can perform value iteration according to the following update equation:
\begin{align}
 & U_{n+1}(i)=\min_{a}\Bigg\{ C_{a}(i)\nonumber \\
 & +\gamma\sum_{i'}\left(\sum_{j'}P{[{g(t)=a(i),g(t+1)=(i',j')}]} \right)U_{n}(i')\Bigg\}\label{eq:approxRatioProofValueIterationNoj}
\end{align}
The difference between (\ref{eq:approxRatioProofValueIteration})
and (\ref{eq:approxRatioProofValueIterationNoj}) is that (\ref{eq:approxRatioProofValueIterationNoj})
does not distinguish the actions and value functions with different
$j$ indices. Recall that for the modified 2-D MDP $\{g(t)\}$, we have $U_{n}(i,j)=U_{n}(i,j')$
for all $n$, $i$ and $j\neq j'$, so the index $j$ can be natually
removed from the value functions. Further, if it is optimal to migrate
at state $(i,j)$ to a state in ring $i'$, it must be optimal to
migrate at states $(i,j)$ for all $j$ to a state (which may not
be the same state) in ring $i'$. The migration cost for different
$j$ are the same because they all follow the shortest path from state
$(i,j)$ to ring $i'$. If it is optimal not to migrate at state $(i,j)$,
then it is optimal not to migrate at states $(i,j)$ for all $j$.
Therefore, we can also remove the $j$ index associated with the actions,
without affecting the value function. It follows that the optimal
policy for $\{g(t)\}$ is equivalent to the optimal policy for $\{d(t)\}$, both bringing the same value functions (discounted sum costs).

\subsection*{Part III -- Error bound for distance-based approximation}

By now, we have shown the upper bound on the discounted sum cost difference
between the original and modified 2-D MDPs $\{e(t)\}$ and $\{g(t)\}$, when both MDPs use the
same policy. We have also shown that the optimal policy for the modified
2-D MDP $\{g(t)\}$ is equivalent to the optimal policy for the distance-based MDP $\{d(t)\}$. Note that
the true optimal cost is obtained by solving for the optimal policy
for the original 2-D MDP $\{e(t)\}$, and the approximate optimal cost is obtained
by applying the optimal policy for $\{d(t)\}$ to $\{e(t)\}$.
In the following, we consider the upper bound on the difference between
the true and approximate optimal discounted sum costs.

We start with a (true) optimal policy $\pi_{\textrm{true}}^{*}$ for $\{e(t)\}$, denote the discounted sum costs from this policy as $V_{\pi_{\textrm{true}}^{*}}(i,j)$.
When using the same policy on $\{g(t)\}$,
the difference between the costs $U_{\pi_{\textrm{true}}^{*}}(i,j)$ and $V_{\pi_{\textrm{true}}^{*}}(i,j)$  satisfies the upper bound given in (\ref{eq:epsilonExpressionW}), i.e.
\begin{equation}
\label{eq:trueOptPolicyCostDiff}
U_{\pi_{\textrm{true}}^{*}}(i,j)-V_{\pi_{\textrm{true}}^{*}}(i,j)\leq\epsilon_{V_{\pi_{\textrm{true}}^{*}}}
\end{equation}

Since $V_{\pi_{\textrm{true}}^{*}}(i,j)$ are the optimal costs, we have 
\begin{align}
\label{eq:sumMaxToOneSlotMax}
 & \max_{i,j,j':(i+1,j)\in\mathcal{N}_{2}\left(i-1,j'\right)}\left|\left(V_{\pi_{\textrm{true}}^{*}}(i\!+\!1,j)\!-\! V_{\pi_{\textrm{true}}^{*}}(i\!-\!1,j')\right)\right|\nonumber \\
 & \leq\max_{x}\left\{ {b}\left(x+2\right)-{b}\left(x\right)\right\} 
\end{align}
because, otherwise, there exists at least one pair of states $(i+1,j)$
and $(i-1,j')$ for which it is beneficial to migrate from state $(i+1,j)$ to state $(i-1,j')$, according to a 2-D extension of (\ref{eq:costRelationship}). 
Note that further migration may occur at state $(i-1,j')$, thus we take the maximum over $x$ in the expression.
The cost after performing such migration is upper bounded by (\ref{eq:sumMaxToOneSlotMax}), which contradicts with the fact that $\pi_{\textrm{true}}^{*}$ is optimal.

Define
\begin{equation}
\epsilon_{c}=\frac{\gamma r\max_{x}\left\{ {b}\left(x+2\right)-{b}\left(x\right)\right\} }{2(1-\gamma)}
\end{equation}
From (\ref{eq:epsilonExpressionW}), (\ref{eq:trueOptPolicyCostDiff}) and (\ref{eq:sumMaxToOneSlotMax}), we have
\begin{equation}
U_{\pi_{\textrm{true}}^{*}}(i,j)-V_{\pi_{\textrm{true}}^{*}}(i,j)\leq\epsilon_{c}\label{eq:approxRatioProofTrueOptError}
\end{equation}

According to the equivalence of $\{g(t)\}$ and $\{d(t)\}$, we know
that the optimal policy $\pi_{\textrm{appr}}^{*}$ of $\{d(t)\}$ is also optimal for $\{g(t)\}$. Hence, we have
\begin{equation}
U_{\pi_{\textrm{appr}}^{*}}(i,j)\leq U_{\pi_{\textrm{true}}^{*}}(i,j)\label{eq:approxRatioProofApproxTrueRelationship}
\end{equation}
because the cost from the optimal policy cannot be higher than the cost from any other policy.

When using the policy $\pi_{\textrm{appr}}^{*}$ on $\{e(t)\}$, we
get costs $V_{\pi_{\textrm{appr}}^{*}}(i,j)$. From (\ref{eq:epsilonExpressionU}),
and because $U_{\pi_{\textrm{appr}}^{*}}(i,j)$ is the optimal cost for $\{g(t)\}$,
we have
\begin{equation}
V_{\pi_{\textrm{appr}}^{*}}(i,j)-U_{\pi_{\textrm{appr}}^{*}}(i,j)\leq\epsilon_{U_{\pi_{\textrm{appr}}^{*}}}\leq\epsilon_{c}\label{eq:approxRatioProofApproxOptError}
\end{equation}

From (\ref{eq:approxRatioProofTrueOptError}), (\ref{eq:approxRatioProofApproxTrueRelationship}),
and (\ref{eq:approxRatioProofApproxOptError}), we get
\begin{equation}
V_{\pi_{\textrm{appr}}^{*}}(i,j)-V_{\pi_{\textrm{true}}^{*}}(i,j)\leq2\epsilon_{c}
\end{equation}
which completes the proof.

\section{ Proof of Theorem \ref{prop:rEstUnbiased} }
\label{sec:proofOfREstUnbiased}

We note that in (\ref{eq:realWorldEst1}), $m_{n}(\tau)$ and $m'_{n}(\tau)$ are both random variables respectively representing the total number of users associated to BS (located in cell) $n$ in slot $\tau$ and the number of users that have moved out of cell $n$ at the end of slot $\tau$. The values of $m_{n}(\tau)$ and $m'_{n}(\tau)$ are random due to the randomness of user mobility. 

According to the mobility model, each user moves out of its current cell at the end of a timeslot with probability $6r$, where we recall that $r\leq \frac{1}{6}$ by definition. Hence, under the condition that $m_{n}(\tau)=k$, $m'_{n}(\tau)$ follows the binomial distribution with parameters $k$ and $6r$. Thus, the conditional probability
\begin{equation}
\Pr\left\{m'_{n}(\tau) =k' |m_{n}(\tau)=k\right\} = \left(\begin{array}{c}
k\\
k'
\end{array}\right)(6r)^{k'}(1-6r)^{k-k'} 
\label{eq:rEstUnbiasedProof0}
\end{equation}
for all $0\leq k' \leq k$.

From the expression of the mean of binomially distributed random variables, we know that the conditional expectation
\begin{equation}
\expect{m'_{n}(\tau)|m_{n}(\tau)}=6r\cdot m_{n}(\tau).
\label{eq:rEstUnbiasedProof0_1}
\end{equation}
Combining (\ref{eq:realWorldEst1})--(\ref{eq:realWorldEst3}), we have
\begin{equation}
\hat{r}=\frac{1}{6 N_{\textrm{BS}} T_w} \sum_{n\in\mathcal{N}_\textrm{BS}} \sum_{\tau=t-T_w}^{t-1}\frac{m'_{n}(\tau)}{m_{n}(\tau)}. 
\label{eq:rEstUnbiasedProof1}
\end{equation}
We then have
\begin{align}
\expect{\hat{r}} & =\expect{ \frac{1}{6 N_{\textrm{BS}} T_w} \sum_{n\in\mathcal{N}_\textrm{BS}} \sum_{\tau=t-T_w}^{t-1}\frac{m'_{n}(\tau)}{m_{n}(\tau)} } \nonumber \\
& =  \frac{1}{6 N_{\textrm{BS}} T_w} \sum_{n\in\mathcal{N}_\textrm{BS}} \sum_{\tau=t-T_w}^{t-1} \expect{\frac{\expect{m'_{n}(\tau)|m_{n}(\tau)}}{m_{n}(\tau)} } \nonumber \\
& =  \frac{1}{6 N_{\textrm{BS}} T_w} \sum_{n\in\mathcal{N}_\textrm{BS}} \sum_{\tau=t-T_w}^{t-1} \expect{\frac{6r\cdot m_{n}(\tau)}{m_{n}(\tau)}}  \nonumber \\
& =  \frac{r}{N_{\textrm{BS}} T_w} N_{\textrm{BS}} T_w  \nonumber \\
& = r \nonumber 
\end{align}
where the second equality follows from the law of iterated expectations.

\section{ Proof of Theorem \ref{prop:rEstVarianceUpperBound} }
\label{sec:proofOfREstVarianceUpperBound}

We define $\mathcal{M} \triangleq \{ m_n(\tau) : \forall n \in \mathcal{N}_{\textrm{BS}}, \tau \in [t-T_w, t-1] \}$ as the set of number of users at all BSs and all timeslots considered in the estimation. We first introduce the following lemma.

\begin{lemma}
\label{lemma:rEstVariance}
Assume that Assumption \ref{assumption:independence} is satisfied and each user follows 2-D random walk (defined at the beginning of Section \ref{section:EstDiscussion}) with parameter $r$.
The variance of estimator $\hat{r}$, under the condition that $\mathcal{M}$ is given, is
\begin{equation}
\mathrm{Var} \{ \hat{r} | \mathcal{M}\} = \frac{r (1-6r)}{6 N^2_{\textrm{BS}} T^2_w} 
\left( \sum_{n\in\mathcal{N}_\textrm{BS}} \sum_{\tau=t-T_w}^{t-1} 
\frac{1}{m_{n}(\tau)} \right)
\label{eq:rEstVarianceProp1}
\end{equation}
where the \emph{conditional variance} $\mathrm{Var} \{ \hat{r} | \mathcal{M}\}$ is defined as
\begin{equation}
\mathrm{Var} \{ \hat{r} | \mathcal{M}\} \triangleq \expect{\hat{r}^2\big| \mathcal{M}} - \left(\expect{\hat{r}| \mathcal{M}}\right)^2 .
\end{equation}
\end{lemma}
\begin{proof}
Taking the conditional expectation on both sides of (\ref{eq:rEstUnbiasedProof1}), we have
\begin{align}
\expect{\hat{r}|\mathcal{M}} & =\expect{\left. \frac{1}{6 N_{\textrm{BS}} T_w} \sum_{n\in\mathcal{N}_\textrm{BS}} \sum_{\tau=t-T_w}^{t-1}\frac{m'_{n}(\tau)}{m_{n}(\tau)} \right|\mathcal{M}} \nonumber \\
& =  \frac{1}{6 N_{\textrm{BS}} T_w} \sum_{n\in\mathcal{N}_\textrm{BS}} \sum_{\tau=t-T_w}^{t-1}\frac{\expect{m'_{n}(\tau)|m_{n}(\tau)}}{m_{n}(\tau)}  \nonumber \\
& =  \frac{1}{6 N_{\textrm{BS}} T_w} \sum_{n\in\mathcal{N}_\textrm{BS}} \sum_{\tau=t-T_w}^{t-1}\frac{6r\cdot m_{n}(\tau)}{m_{n}(\tau)}  \nonumber \\
& =  \frac{r}{N_{\textrm{BS}} T_w} N_{\textrm{BS}} T_w  \nonumber \\
& = r 
\label{eq:rEstVarianceProofCondExpect}
\end{align}
where the second equality is because $m'_n(\tau)$ is independent of $m_n(\tilde{\tau})$, $m_{\tilde{n}}(\tau)$, and $m_{\tilde{n}}(\tilde{\tau})$ (where $\tilde{n} \neq n$ and $\tilde{\tau} \neq \tau$) when $m_n(\tau)$ is given, according to Assumption~\ref{assumption:independence}.
We thus have $(\expect{\hat{r} | \mathcal{M}})^2 = r^2$. 

We focus on evaluating $\expect{\hat{r}^2| \mathcal{M}}$ in the following.
From (\ref{eq:rEstUnbiasedProof1}), we have
\begin{align}
\hat{r}^2 & =\frac{1}{36 N^2_{\textrm{BS}} T^2_w}
\left( \sum_{n\in\mathcal{N}_\textrm{BS}} \sum_{\tau=t-T_w}^{t-1}\frac{m'_{n}(\tau)}{m_{n}(\tau)} \right)^2 \nonumber \\
& = \frac{1}{36 N^2_{\textrm{BS}} T^2_w} \left(
\sum_{n\in\mathcal{N}_\textrm{BS}} \sum_{\tau=t-T_w}^{t-1}\frac{(m'_{n}(\tau))^2}{(m_{n}(\tau))^2} \right. \\
& \quad\quad \left. + \sum_{\substack{n_1, n_2\in\mathcal{N}_\textrm{BS}; \\
\tau_1, \tau_2 \in [t-T_w, t-1]; \\
n_1 \neq n_2 \textrm{ and/or } \tau_1 \neq \tau_2}} \frac{m'_{n_1}(\tau_1)}{m_{n_1}(\tau_1)}\cdot \frac{m'_{n_2}(\tau_2)}{m_{n_2}(\tau_2)} \right).
\label{eq:rEstVarianceProof1}
\end{align}
We now consider the two parts in 
(\ref{eq:rEstVarianceProof1}). From the proof of Theorem \ref{prop:rEstUnbiased}, we know that $m'_n(\tau)$ follows the binomial distribution when $m_n(\tau)$ is given. Further,  when $m_n(\tau)$ is given, $m'_n(\tau)$ is independent of $m_n(\tilde{\tau})$, $m_{\tilde{n}}(\tau)$, and $m_{\tilde{n}}(\tilde{\tau})$ (where $\tilde{n} \neq n$ and $\tilde{\tau} \neq \tau$), according to Assumption \ref{assumption:independence}.
Thus, we have
\begin{align}
\expect{\frac{(m'_{n}(\tau))^2}{(m_{n}(\tau))^2} \Bigg| \mathcal{M}}
& = \frac{\expect{(m'_{n}(\tau))^2\big| m_{n}(\tau)}}
{(m_{n}(\tau))^2}  \nonumber \\
& = \frac{m_{n}(\tau) \cdot 6r \cdot (1-6r) + 36r^2 (m_{n}(\tau))^2}
{(m_{n}(\tau))^2} \nonumber \\
& = \frac{6r \cdot (1-6r)}{m_{n}(\tau)}  + 36r^2  
\label{eq:rEstVarianceProof2}
\end{align}
where the second equality is a known result for binomially distributed random variables.
We also have
\begin{align}
& \expect{\frac{m'_{n_1}(\tau_1)}{m_{n_1}(\tau_1)}\cdot \frac{m'_{n_2}(\tau_2)}{m_{n_2}(\tau_2)} \Bigg| \mathcal{M}} \nonumber \\
& = \expect{\frac{m'_{n_1}(\tau_1)}{m_{n_1}(\tau_1)} \Bigg| m_{n_1}(\tau_1)}\cdot  \expect{\frac{m'_{n_2}(\tau_2)}{m_{n_2}(\tau_2)} \Bigg| m_{n_2}(\tau_2)} \nonumber \\
& = \frac{\expect{m'_{n_1}(\tau_1) | m_{n_1}(\tau_1)}}{m_{n_1}(\tau_1)} \cdot  \frac{\expect{m'_{n_2}(\tau_2)| m_{n_2}(\tau_2)} }{m_{n_2}(\tau_2)} \nonumber \\
& = 36r^2 
\label{eq:rEstVarianceProof3}
\end{align}
for $n_1 \neq n_2$ and/or $\tau_1 \neq \tau_2$, 
where the first equality follows from the fact that $m'_{n_1}(\tau_1)$ and $m'_{n_2}(\tau_2)$ are independent when $m_{n_1}(\tau_1)$ and $m_{n_2}(\tau_2)$ are given (according to Assumption \ref{assumption:independence}), the last equality follows from (\ref{eq:rEstUnbiasedProof0_1}).

We now take the conditional expectation on both sides of (\ref{eq:rEstVarianceProof1}), and substitute corresponding terms with (\ref{eq:rEstVarianceProof2}) and (\ref{eq:rEstVarianceProof3}). This yields
\begin{align}
& \expect{\hat{r}^2\big| \mathcal{M}} \nonumber \\
& = \frac{1}{36 N^2_{\textrm{BS}} T^2_w} \left(
\sum_{n\in\mathcal{N}_\textrm{BS}} \sum_{\tau=t-T_w}^{t-1} 
\expect{\frac{(m'_{n}(\tau))^2}{(m_{n}(\tau))^2} \Bigg| \mathcal{M}} \right. \nonumber \\
& \quad{}\quad{} \left. + \sum_{\substack{n_1, n_2\in\mathcal{N}_\textrm{BS}; \\
\tau_1, \tau_2 \in [t-T_w, t-1]; \\
n_1 \neq n_2 \textrm{ and/or } \tau_1 \neq \tau_2}} 
\expect{\frac{m'_{n_1}(\tau_1)}{m_{n_1}(\tau_1)}\cdot \frac{m'_{n_2}(\tau_2)}{m_{n_2}(\tau_2)} \Bigg| \mathcal{M}}
 \right) \nonumber \\
& = \frac{1}{36 N^2_{\textrm{BS}} T^2_w} \left(
\sum_{n\in\mathcal{N}_\textrm{BS}} \sum_{\tau=t-T_w}^{t-1} 
\left(\frac{6r \cdot (1-6r)}{m_{n}(\tau)}  + 36r^2\right) \right. \\
& \quad\quad \left. + \sum_{\substack{n_1, n_2\in\mathcal{N}_\textrm{BS}; \\
\tau_1, \tau_2 \in [t-T_w, t-1]; \\
n_1 \neq n_2 \textrm{ and/or } \tau_1 \neq \tau_2}} 
36r^2 \right) \nonumber \\
& = \frac{1}{36 N^2_{\textrm{BS}} T^2_w} \left( 
6r \cdot (1-6r) \cdot \left( \sum_{n\in\mathcal{N}_\textrm{BS}} \sum_{\tau=t-T_w}^{t-1} 
\frac{1}{m_{n}(\tau)} \right) \right. \\
& \quad\quad \left. +N^2_{\textrm{BS}} T^2_w \cdot 36r^2
\right) \nonumber \\
& = \frac{r (1-6r)}{6 N^2_{\textrm{BS}} T^2_w} 
\left( \sum_{n\in\mathcal{N}_\textrm{BS}} \sum_{\tau=t-T_w}^{t-1} 
\frac{1}{m_{n}(\tau)} \right)
+r^2.
\end{align}
Subtracting $(\expect{\hat{r} | \mathcal{M}})^2 = r^2$ from the above yields the result.
\end{proof}

Lemma \ref{lemma:rEstVariance} gives the conditional variance of estimator $\hat{r}$ when the values of $m_n(\tau)$ are given. It is not straightforward to remove the condition, because it is hard to find the stationary distribution of user locations in a hexagonal 2-D mobility model with a finite number of cells. We note that the set of $m_n(\tau)$ values represents the set of samples in our estimation problem. In standard estimation problems, the sample size is usually deterministic, while it is random in our problem due to random user locations. This causes the difficulty in finding the unconditional variance of our estimator.

However, Lemma \ref{lemma:rEstVariance} is important because it gives us a sense on how large the gap between $\hat{r}$ and $r$ is, provided that each user precisely follows the random walk mobility model. It also enables us to proof Theorem~\ref{prop:rEstVarianceUpperBound}.

\begin{proof} (\textbf{Theorem  \ref{prop:rEstVarianceUpperBound}})
As discussed in Section \ref{sub:overallProcedureTraces},
we assume that $m_n (\tau) \neq 0$ for all $n$ and $\tau$. Thus, we always have $\frac{1}{m_{n}(\tau)} \leq 1$ (since $m_{n}(\tau)$ is a positive integer) and
\begin{equation}
\sum_{n\in\mathcal{N}_\textrm{BS}} \sum_{\tau=t-T_w}^{t-1} 
\frac{1}{m_{n}(\tau)} \leq N_{\textrm{BS}} T_w. 
\label{eq:rEstVarianceUpperBoundProof1}
\end{equation}
We also have the following bound:
\begin{equation}
r(1-6r) \leq \frac{1}{24}
\label{eq:rEstVarianceUpperBoundProof2}
\end{equation}
for any $r \in [0, \frac{1}{6}]$.
The law of total variance gives
\begin{equation}
\mathrm{Var} \{ \hat{r}\}= \expect{\mathrm{Var} \{ \hat{r} | \mathcal{M}\}} + \mathrm{Var} \{ \expect{ \hat{r} | \mathcal{M}}\}
\nonumber
\end{equation}
According to (\ref{eq:rEstVarianceProofCondExpect}), $\expect{ \hat{r} | \mathcal{M}}=r$ which is a constant, thus $\mathrm{Var} \{ \expect{ \hat{r} | \mathcal{M}}\}=0$.
Therefore, we have
\begin{equation}
\mathrm{Var} \{ \hat{r}\}
= \expect{\mathrm{Var} \{ \hat{r} | \mathcal{M}\}}
\leq \max_\mathcal{M} \mathrm{Var} \{ \hat{r} | \mathcal{M}\}
\leq \frac{1}{144 N_{\textrm{BS}} T_w}
\nonumber
\end{equation}
where the last inequality follows from substituting (\ref{eq:rEstVarianceUpperBoundProof1}) and (\ref{eq:rEstVarianceUpperBoundProof2}) into (\ref{eq:rEstVarianceProp1}).
\end{proof}

\section{Proof of Theorem \ref{prop:approxErrorBoundAddConstraints}}
\label{sec:ProofOfApproxErrorBoundAddConstraints}

Recall that for the original distance-based MDP $\{d(t)\}$ and its discounted sum cost, we have the following according to (\ref{eq:balanceDiscountedSumCost}):
\begin{equation}
V^*(d) = C_{a^*}(d) + \gamma \sum_{d(t+1)=0}^{N}  P\left[d'(t),d(t+1)\right]  V^*(d(t+1)) \label{eq:prop:approxErrorBoundAddConstraints:proofBalanceV}
\end{equation}
where $d'(t)=a^*(d)$. The transition probability $P\left[d'(t),d(t+1)\right]$ is defined according to Fig.~\ref{fig:states1D}. We always have $P\left[d'(t),d(t+1)\right]=0$ when $|d(t+1)-d'(t)|>1$.

We note that the distance-based MDP following the actual (randomized) action $a(d)$ is equivalent to an MDP following the optimal action $a^*(d)$ with modified transition probabilities. Let $\{\xi(t)\}$ denote such a modified MDP. There is a one-to-one mapping between the states in $\{\xi(t)\}$ and the states in $\{d(t)\}$, but the transition probabilities of these two MDPs are different. For convenience of comparison later, we use $\{d(t)\}$ to denote the states of $\{\xi(t)\}$ and use $P_\xi[\cdot,\cdot]$ to denote the transition probability of the MDP $\{\xi(t)\}$. We have
\begin{align}
& P_\xi\left[d'(t),d(t+1)\right] \nonumber \\
& =\!
\begin{cases}
(1-p_{0})\sum_{k\leq\Delta_t}\alpha_{k}+q\alpha_{\Delta_t+1}, \!&\! \textrm{if }d(t\!+\!1)\!=\!0\\
p_{0}\!\sum_{k\leq\Delta_t-1}\alpha_{k}\!+\!\nu\alpha_{\Delta_t}\!+\!q\alpha_{\Delta_t+1}, \!&\! \textrm{if }d(t\!+\!1)\!=\!1\\
p\alpha_{\Delta_t-1}+\nu\alpha_{\Delta_t}+q\alpha_{\Delta_t+1},  \!&\! \textrm{if } d(t\!+\!1)\!\in\![2, N\!\!-\!3] \\
p\alpha_{\Delta_t-1}\!+\!\nu\alpha_{\Delta_t}\!+\!q\sum_{k\geq\Delta_t+1}\alpha_{k},\! &\! \textrm{if } d(t\!+\!1)\!=\!N\!\!-\!2 \\
p\alpha_{\Delta_t-1}+\nu\sum_{k\geq\Delta_t}\alpha_{k}, \!&\! \textrm{if }d(t\!+\!1)\!=\!N\!\!-\!1 \\
p\alpha_{\Delta_t-1}, \!&\! \textrm{if }d(t\!+\!1)\!=\!N
\end{cases}
\label{eq:prop:approxErrorBoundAddConstraints:proof1}
\end{align}
where $\nu \triangleq 1-p-q$ and $\Delta_t \triangleq d(t+1)-d'(t)$. This transition probability expression is from the definition of $\alpha_k$ and the original distance-based MDP definition in Fig.~\ref{fig:states1D}. The above expression holds for $N\geq 5$; for smaller values of $N$, similar expressions can be derived and we omit the details.

By definition, $\tilde{V}(d)$ is equal to the discounted sum cost of this modified MDP $\{\xi(t)\}$. We have the following balance equation:
\begin{equation}
\tilde{V}(d) = C_{a^*}(d) + \gamma \sum_{d(t+1)=0}^{N}  P_\xi\left[d'(t),d(t+1)\right]  \tilde{V}(d(t+1)). \label{eq:prop:approxErrorBoundAddConstraints:proofBalanceU}
\end{equation}

Similar to the first part of the proof of Theorem~\ref{prop:approxErrorBound}, we assume that $\left|\tilde{V}(d)-V^*(d)\right|\leq\epsilon$ for all $d$ and the value of $\epsilon$ will be determined later. Let $\Psi(d(t+1)) \triangleq P_\xi\left[d'(t),d(t+1)\right]  - P\left[d'(t),d(t+1)\right]$. We have
\begin{align}
& \left|\tilde{V}(d)-V^*(d)\right| \nonumber \\
& = \gamma \Bigg| \sum_{d(t+1)=0}^{N} \!\!\!\!\! P_\xi\left[d'(t),d(t+1)\right]  \left(\tilde{V}(d(t+1)) - V^*(d(t+1)) \right) \nonumber \\
& \quad \quad + \sum_{d(t+1)=0}^{N}  \Psi(d(t+1)) V^*(d(t+1))   \Bigg| \nonumber \\
& \leq \gamma\epsilon + \gamma  \Bigg| \sum_{d(t+1)=0}^{N} \Psi(d(t+1)) V^*(d(t+1))   \Bigg| \label{eq:prop:approxErrorBoundAddConstraints:proof2}
\end{align}

We consider the second term in (\ref{eq:prop:approxErrorBoundAddConstraints:proof2}) in the following.
We first note that
\begin{align}
\left|V^*\left(\tilde{d}\right) - V^*(d)\right| \leq \max_x \left\{b\left(x + \left|\tilde{d}-d\right|\right) - b(x) \right\}  \label{eq:prop:approxErrorBoundAddConstraints:proof3}
\end{align}
for any $d$ and $\tilde{d} > d$, because otherwise, there exist some $\tilde{d}$ and $d$, for which (\ref{eq:prop:approxErrorBoundAddConstraints:proof3}) does not hold, where one can choose a new action that migrates from $\tilde{d}$ to $d$. When using this new action, (\ref{eq:prop:approxErrorBoundAddConstraints:proof3}) is satisfied (further migration may occur at state $d$, thus we take the maximum over $x$ in the expression) and the discounted sum cost $V^*(d)$ is reduced. This contradicts with the Bellman's equation in (\ref{eq:bellman}) since $V^*(d)$ is for the optimal policy.

We also note that 
\begin{align*}
& \sum_{d(t+1)=0}^{N} \Psi(d(t+1)) \\
& = \sum_{d(t+1)=0}^{N} P_\xi\left[d'(t),d(t+1)\right]  - \sum_{d(t+1)=0}^{N} P\left[d'(t),d(t+1)\right] \\
& = 1-1 = 0.
\end{align*}
Therefore, the sum of all positive $\Psi(d(t+1))$ is equal to the sum of all negative $\Psi(d(t+1))$ in (\ref{eq:prop:approxErrorBoundAddConstraints:proof2}).

Because $\alpha_k = 0$ for all $|k|>K$, according to the definition of $P_\xi\left[d'(t),d(t+1)\right]$ in (\ref{eq:prop:approxErrorBoundAddConstraints:proof1}), we have $P_\xi\left[d'(t),d(t+1)\right]=0$ when $\Delta_t + 1 < -K \leq 0$ or $\Delta_t - 1 > K \geq 0$.
It follows that $P_\xi\left[d'(t),d(t+1)\right]$ can be non-zero only when $|\Delta_t | \leq K+1$.
For the original distance-based MDP, its transition probability $P\left[d'(t),d(t+1)\right]$ can be non-zero only when $|\Delta_t | \leq 1$ according to the definition.
It follows that $\Psi(d(t+1))$ can only be negative when $|\Delta_t | \leq 1$.

Therefore, for any $d$ and $\tilde{d}$ such that $\Psi(d(t+1))>0$ and $\Psi(\tilde{d}(t+1))<0$, we must have 
\begin{equation}
\left| d(t+1) - \tilde{d}(t+1)  \right| \leq K+2. \label{eq:prop:approxErrorBoundAddConstraints:proof4}
\end{equation}
Expanding the sum in (\ref{eq:prop:approxErrorBoundAddConstraints:proof2}) into separate positive and negative terms, we have
\begin{align}
& \left| \sum_{d(t+1)=0}^{N} \Psi(d(t+1)) V^*(d(t+1)) \right| \nonumber \\
& = \Bigg| \sum_{d(t+1):\Psi(d(t+1))>0} \Psi(d(t+1)) V^*(d(t+1)) \nonumber \\
& \quad - \sum_{\tilde{d}(t+1):\Psi\left(\tilde{d}(t+1)\right)<0} \left|\Psi\left(\tilde{d}(t+1)\right)\right| \cdot V\left(\tilde{d}(t+1)\right)  \Bigg| \nonumber \\
& \leq \max_x \left\{b\left(x + K+2\right) - b(x) \right\} \label{eq:prop:approxErrorBoundAddConstraints:proof5}
\end{align}
where the last inequality is from (\ref{eq:prop:approxErrorBoundAddConstraints:proof3}), (\ref{eq:prop:approxErrorBoundAddConstraints:proof4}), and the fact that 
\begin{align*}
& \sum_{d(t+1):\Psi(d(t+1))>0} \Psi(d(t+1)) \\
& \quad\quad = \sum_{\tilde{d}(t+1):\Psi\left(\tilde{d}(t+1)\right)<0} \left|\Psi\left(\tilde{d}(t+1)\right)\right| \leq 1.
\end{align*}

Substituting (\ref{eq:prop:approxErrorBoundAddConstraints:proof5}) into (\ref{eq:prop:approxErrorBoundAddConstraints:proof2}), we get
\begin{align*}
\left|\tilde{V}(d)-V^*(d)\right| \leq  \gamma\epsilon + \gamma  \max_x \left\{b\left(x + K+2\right) - b(x) \right\}
\end{align*}

Because $0 < \gamma < 1$, $V^*(d)$ and $\tilde{V}(d)$ both converge to some fixed value (for each $d$) when iterating according to (\ref{eq:prop:approxErrorBoundAddConstraints:proofBalanceV}) and (\ref{eq:prop:approxErrorBoundAddConstraints:proofBalanceU}), respectively~\cite[Chapter 6]{puterman2009markov}. Hence, $\left|\tilde{V}(d)-V^*(d)\right|$ also converges to some fixed value (for each $d$). We can therefore solve $\epsilon$ as
\begin{equation*}
\epsilon = \frac{\gamma  \max_x \left\{b\left(x + K+2\right) - b(x) \right\}}{1-\gamma}
\end{equation*}
which completes the proof.

\section{ Additional Simulation Results with Real-World Traces}
\label{sec:AdditionalSimulationResults}

To consider the overall performance under different parameter settings, we denote the cost of the proposed method as $C$ and the cost of the baseline method under comparison as $C_0$, and we define the \emph{cost reduction} as $(C_0-C)/C_0$. Figs.~\ref{fig:simAvrHexNonConst}--\ref{fig:simAvrRealConst} show the cost reductions (averaged over the entire day) under different parameter settings. A positive cost reduction indicates that the proposed method performs better than the baseline, and the reverse is true for a negative cost reduction.

We can see that under the default number of ESs and capacity limit at each ES, the proposed approach is beneficial with average cost reductions of up to $44\%$ compared to the never/always-migrate or myopic policies. 

We also see that the cost reductions compared to never-migrate and myopic policies become small in the case where either the number of ESs is small or the capacity of each ES is low. In this case, it is hardly possible to migrate the service to a better location because the action space for migration is very small. 
Therefore, the proposed approach gives a similar cost as baseline policies, but still outperforms them as indicated by a positive cost reduction on average. The cost reduction compared to the always-migrate policy becomes slightly higher in the case of small action space, because the always-migrate policy always incurs migration cost even though the benefit of such migration is not obvious.

It is also interesting to see that while the average performance of the proposed approach is the best in almost all cases, there exist some instances in which the error bar (showing the standard deviation of the cost reductions computed on instantaneous costs) in Figs.~\ref{fig:simAvrHexNonConst}--\ref{fig:simAvrRealConst} goes below zero, indicating that the instantaneous cost of the proposed approach may be higher than some baseline approaches in some instances. This is possible because the proposed approach aims at minimizing the discounted sum cost defined in (\ref{eq:discountedSumCost}) and we use $\gamma = 0.9$ in the simulations. Thus, it may not minimize the instantaneous costs in all timeslots. The myopic policy, on the other hand, aims at minimizing the instantaneous cost. We note that setting $\gamma = 0$ makes our algorithm the same as the myopic policy. Hence, our algorithm is also capable of minimizing the instantaneous cost if the application desires so. However, in many practical scenarios, one would like to achieve a balance between the instantaneous cost and the long-term average cost, so that the service performance is reasonably good both instantaneously and on average. The discount factor  $\gamma$ acts as a control knob that adjusts this balance.

\begin{figure*}
\center{
\subfigure[]{\includegraphics[width=0.35\textwidth]{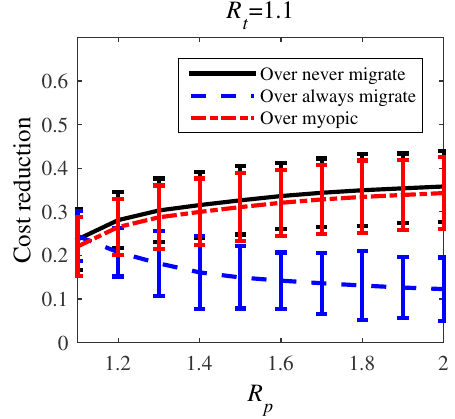}}
\subfigure[]{\includegraphics[width=0.35\textwidth]{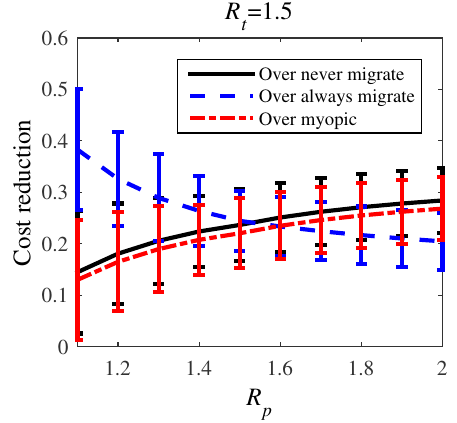}}
}

\center{
\subfigure[]{\includegraphics[width=0.35\textwidth]{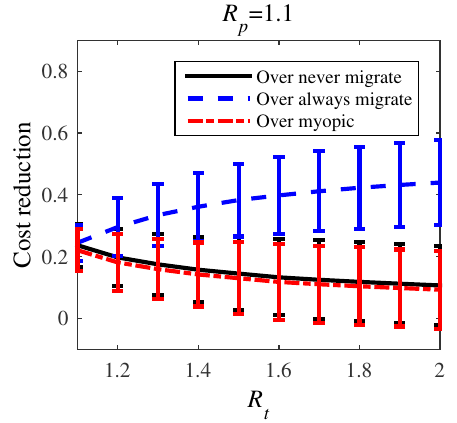}}
\subfigure[]{\includegraphics[width=0.35\textwidth]{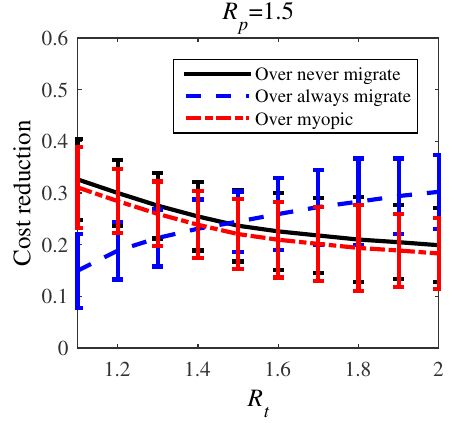}}
}

\center{
\subfigure[]{\includegraphics[width=0.35\textwidth]{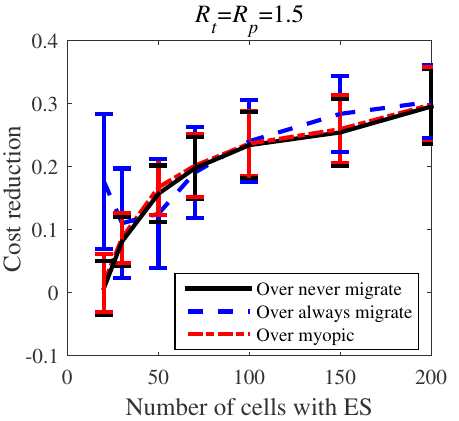}}
\subfigure[]{\includegraphics[width=0.35\textwidth]{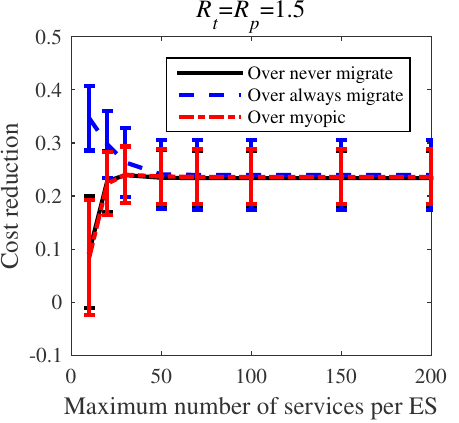}}
}

\protect\caption{(Hexagon, non-constant cost) Cost reduction (averaged over the entire day) compared to alternative policies in trace-driven simulation, the error bars denote the standard deviation (where we regard the cost reduction of instantaneous cost at different time of the day as samples): (a)--(b) cost reduction vs. different $R_t$,  (c)--(d) cost reduction vs. different $R_p$, (e) cost reduction vs. different number of cells with ES, (f) cost reduction vs. different capacity limit of each ES (expressed as the maximum number of services allowed per ES).}
\label{fig:simAvrHexNonConst} 
\end{figure*}

\begin{figure*}
\center{
\subfigure[]{\includegraphics[width=0.35\textwidth]{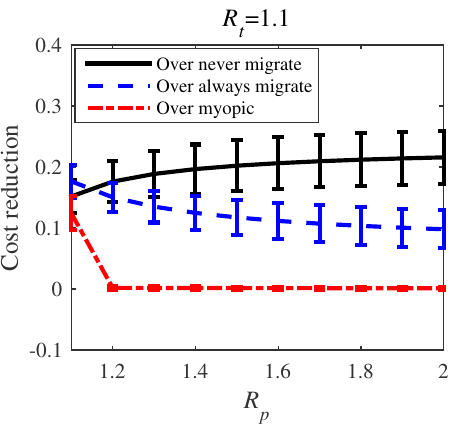}}
\subfigure[]{\includegraphics[width=0.35\textwidth]{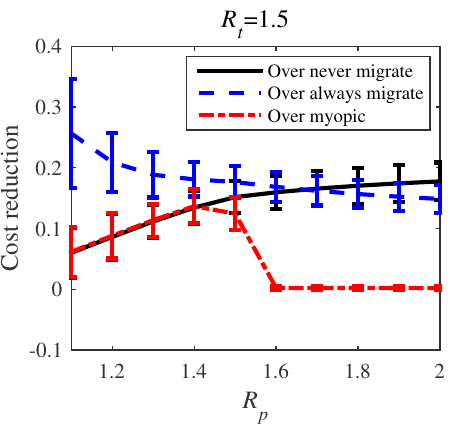}}
}

\center{
\subfigure[]{\includegraphics[width=0.35\textwidth]{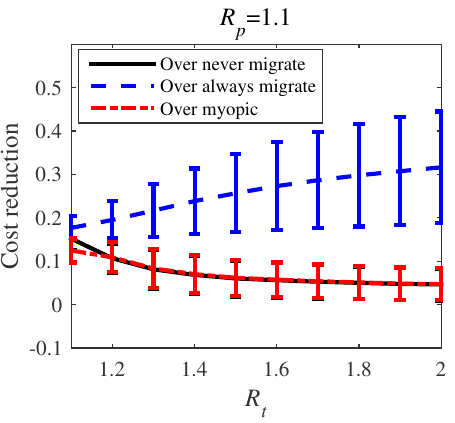}}
\subfigure[]{\includegraphics[width=0.35\textwidth]{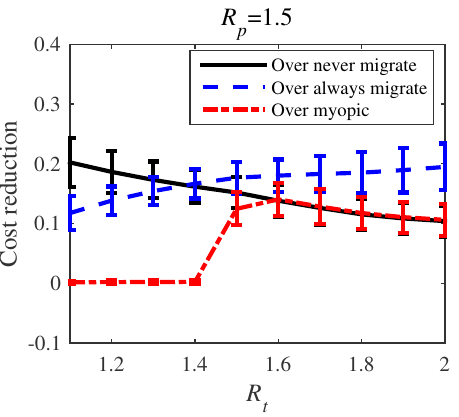}}
}

\center{
\subfigure[]{\includegraphics[width=0.35\textwidth]{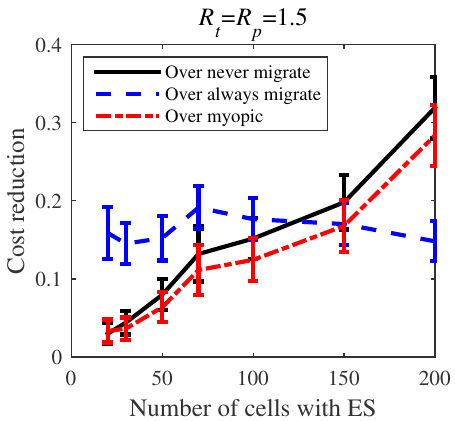}}
\subfigure[]{\includegraphics[width=0.35\textwidth]{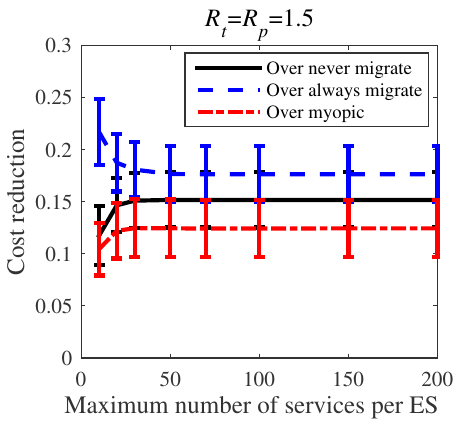}}
}

\protect\caption{(Hexagon, constant cost) Cost reduction (averaged over the entire day) compared to alternative policies in trace-driven simulation, the error bars denote the standard deviation (where we regard the cost reduction of instantaneous cost at different time of the day as samples): (a)--(b) cost reduction vs. different $R_t$,  (c)--(d) cost reduction vs. different $R_p$, (e) cost reduction vs. different number of cells with ES, (f) cost reduction vs. different capacity limit of each ES (expressed as the maximum number of services allowed per ES).}
\label{fig:simAvrHexConst} 
\end{figure*}

\begin{figure*}
\center{
\subfigure[]{\includegraphics[width=0.35\textwidth]{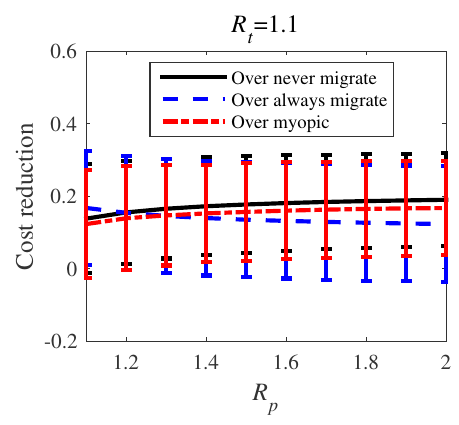}}
\subfigure[]{\includegraphics[width=0.35\textwidth]{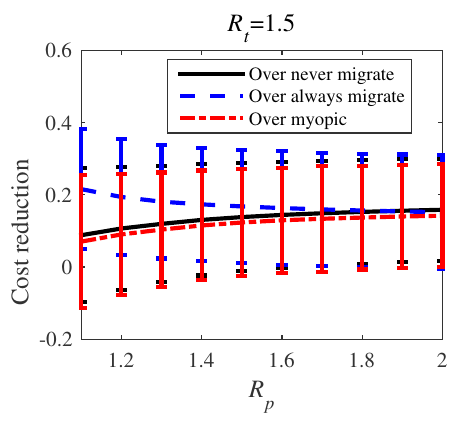}}
}

\center{
\subfigure[]{\includegraphics[width=0.35\textwidth]{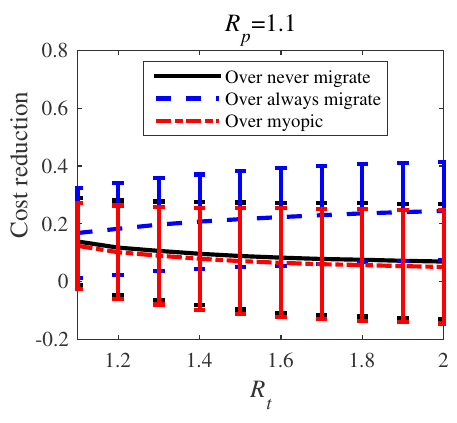}}
\subfigure[]{\includegraphics[width=0.35\textwidth]{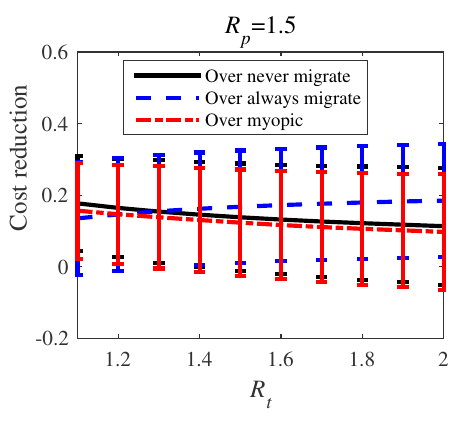}}
}

\center{
\subfigure[]{\includegraphics[width=0.35\textwidth]{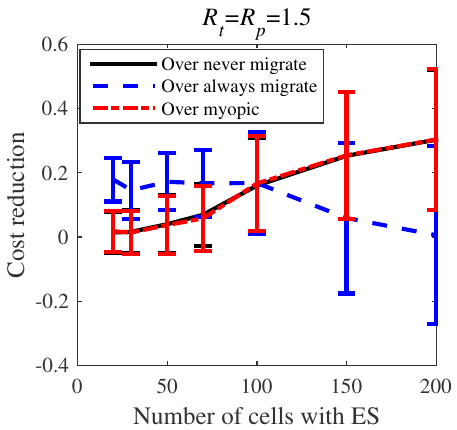}}
\subfigure[]{\includegraphics[width=0.35\textwidth]{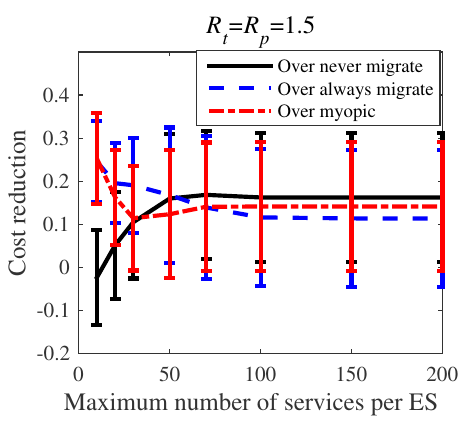}}
}

\protect\caption{(Real, non-constant cost) Cost reduction (averaged over the entire day) compared to alternative policies in trace-driven simulation, the error bars denote the standard deviation (where we regard the cost reduction of instantaneous cost at different time of the day as samples): (a)--(b) cost reduction vs. different $R_t$,  (c)--(d) cost reduction vs. different $R_p$, (e) cost reduction vs. different number of cells with ES, (f) cost reduction vs. different capacity limit of each ES (expressed as the maximum number of services allowed per ES).}
\label{fig:simAvrRealNonConst} 
\end{figure*}

\begin{figure*}
\center{
\subfigure[]{\includegraphics[width=0.35\textwidth]{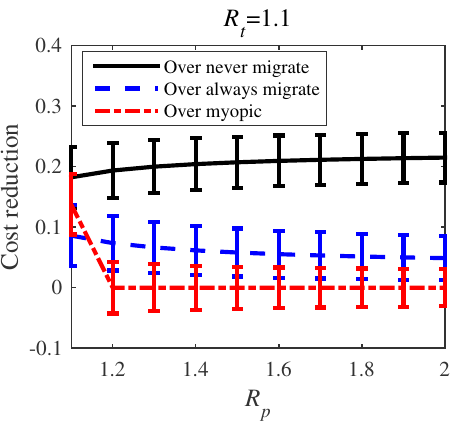}}
\subfigure[]{\includegraphics[width=0.35\textwidth]{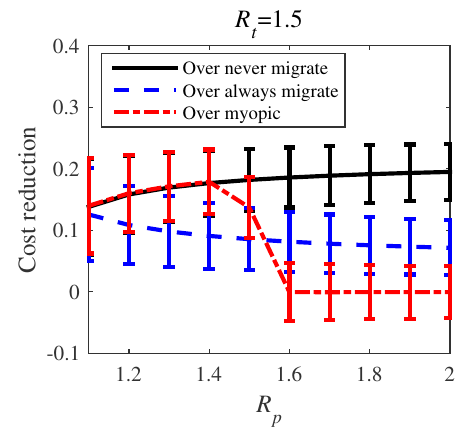}}
}

\center{
\subfigure[]{\includegraphics[width=0.35\textwidth]{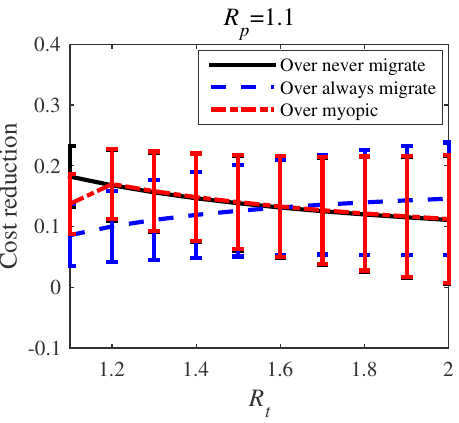}}
\subfigure[]{\includegraphics[width=0.35\textwidth]{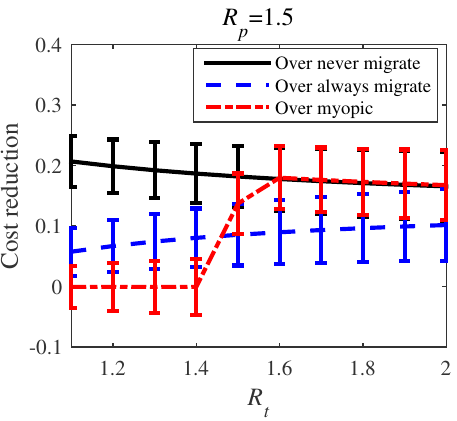}}
}

\center{
\subfigure[]{\includegraphics[width=0.35\textwidth]{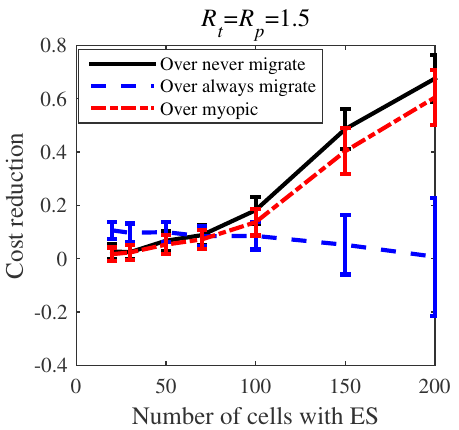}}
\subfigure[]{\includegraphics[width=0.35\textwidth]{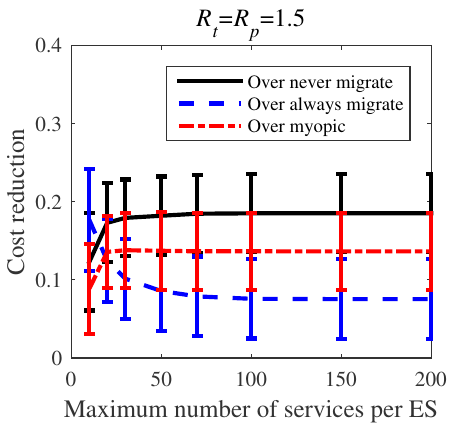}}
}

\protect\caption{(Real, constant cost) Cost reduction (averaged over the entire day) compared to alternative policies in trace-driven simulation, the error bars denote the standard deviation (where we regard the cost reduction of instantaneous cost at different time of the day as samples): (a)--(b) cost reduction vs. different $R_t$,  (c)--(d) cost reduction vs. different $R_p$, (e) cost reduction vs. different number of cells with ES, (f) cost reduction vs. different capacity limit of each ES (expressed as the maximum number of services allowed per ES).}
\label{fig:simAvrRealConst} 
\end{figure*}

\end{document}